\documentclass[aps,onecolumn,showpacs,superscriptaddress,a4paper]{revtex4}
\usepackage{graphicx}
\usepackage{hyperref}
\usepackage{color}
\usepackage{amsmath}
\usepackage{amsthm}
\usepackage{amssymb}

\theoremstyle{plain}
\newtheorem{theorem}{Theorem}
\newtheorem{lemma}{Lemma}

\newtheorem{algorithm}{Algorithm}
\newtheorem{definition}{Definition}
\theoremstyle{definition}
\newtheorem{example}{Example}

\DeclareMathOperator{\B}{\mathcal{B}} 
\DeclareMathOperator{\CPT}{CPT} 
\DeclareMathOperator{\Tr}{Tr} 
\DeclareMathOperator{\id}{id} 

\def\lsim{\mathrel{\rlap{\lower4pt\hbox{\hskip1pt$\sim$}}
    \raise1pt\hbox{$<$}}}                
\def\gsim{\mathrel{\rlap{\lower4pt\hbox{\hskip1pt$\sim$}}
    \raise1pt\hbox{$>$}}}                
\newcommand{\ket}[1]{\left| #1\right\rangle}        
\newcommand{\bra}[1]{\left\langle #1\right|}        
\newcommand{\ketbra}[2]{| #1 \rangle\langle #2 |} 
\newcommand{\caseone}{(i)}
\newcommand{\casetwo}{(ii)}
\newcommand{\fail}{p_\mathrm{fail}}

\begin{document}

\title{Fault-ignorant Quantum Search}
\author{P\'eter Vrana}\thanks{vranap@math.bme.hu}
\affiliation{Institute for Theoretical Physics, ETH Z\"urich, 8093 Z\"urich, Switzerland}
\author{David Reeb}\thanks{david.reeb@tum.de}
\affiliation{Department of Mathematics, Technische Universit\"at M\"unchen, 85748 Garching, Germany}
\author{Daniel Reitzner}\thanks{reitzner@savba.sk}
\affiliation{Department of Mathematics, Technische Universit\"at M\"unchen, 85748 Garching, Germany}
\affiliation{Institute of Physics, Slovak Academy of Sciences, 845 11 Bratislava, Slovakia}
\author{Michael M.\ Wolf} \thanks{m.wolf@tum.de}
\affiliation{Department of Mathematics, Technische Universit\"at M\"unchen, 85748 Garching, Germany}
\begin{abstract}
We investigate the problem of quantum searching on a noisy quantum computer.
Taking a \emph{fault-ignorant} approach, we analyze quantum algorithms that solve the task for various different noise strengths,
which are possibly unknown beforehand. We prove lower bounds on the runtime of such algorithms and thereby find that the quadratic speedup is necessarily lost (in our noise models). However, for low but constant noise levels the algorithms we provide (based on Grover's algorithm) still outperform the best noiseless classical search algorithm.
\end{abstract}

\maketitle


\vspace{-0.3cm}
\section{Introduction}\label{sec:intro}
Since the inception of quantum computing \cite{nielsenchuang}, a large effort has been devoted to making quantum computers function in a noisy environment and securing them against imperfections in the physical setup itself. The theoretical literature offers several ways to cope with such errors. The leading idea is to encode quantum states into a larger system \cite{Shor,Steane} such that noise hits can be recognized and subsequently corrected. A quantum computation can then be performed in a fault-tolerant way directly on these encoding systems \cite{shorfaulttolerant,klz98b}, and nested levels of error-correction can make the computation error-free. The last statement assumes that the initial error rate is not too high and that the error hits are not too correlated, e.g.~occur locally. Beyond these assumptions, quantum error-correction schemes require significant resources in terms of circuit size and experimental control \cite{preskill,gottesman}. Other approaches use decoherence-free subspaces, employing for the quantum computation those parts of a system's Hilbert space which are relatively unaffected by the noise \cite{Zanardi,Lidar}. The latter approach works only in more limited circumstances and requires detailed knowledge of the noise.

Here, we follow a different idea which tries to avoid the disadvantages just mentioned. Rather than devoting large efforts to eliminating the errors at all costs, we accept them in the computation and try to design algorithms that eventually still find the desired result. The spatial size of the quantum circuit should however not be enlarged much beyond the noiseless version of the algorithm; one may e.g.\ allow only a number of extra qubits that remains bounded by a constant as the problem size becomes large \cite{gottesman} (whereas it seems reasonable to allow for an exponentially large noiseless classical memory).

Furthermore, the algorithms should find the desired result under any level of background noise --- this level may actually be unknown to the algorithm, hence the term \emph{fault-ignorant computation}. The algorithm is, however, allowed to take more time the larger the actual noise gets. In this sense, we are trading spatial resources (circuit size) for temporal resources (runtime). Still, when the actual noise level is low (but constant in the problem size), we want our algorithms to produce the desired result faster than any classical algorithm even in a noiseless environment --- this indeed will be the case for the explicit algorithms we present. We describe the fault-ignorant approach and these desiderata in more detail in Section \ref{moreformaldescription}.

While under noise the resulting fault-ignorant algorithms may not give the full quantum speedup for large problem sizes, they may still be useful for initial and medium-term realizations of quantum computers, in particular in the non-scalable low-qubit number regime which does not allow for full-blown quantum error correction schemes.

\bigskip

In this paper, we analyze the above \emph{fault-ignorant idea} for the example of quantum search on a noisy quantum computer. The unstructured search problem, i.e.~the search for a marked item in an unordered list with oracle access, can be solved on a noiseless quantum computer by Grover's algorithm  quadratically faster than by any classical algorithm \cite{TheGrover,nielsenchuang}. This quantum speedup is optimal \cite{BBBV,BBHT,Zalka}.

Good algorithms or optimality issues for the noisy case are however much less clear. Many studies in this direction investigate how (different models of) background noise or a faulty oracle affect specifically Grover's algorithm (e.g.~\cite{Lang,Hsieh,Shapira,Azuma,PabloNorman,SBW}), and it is often found that the quadratic speedup persists as long as the noise stays below a certain threshold depending on the database size. Only the paper \cite{RegevSchiff} by Regev and Schiff gives, for a specific model of a faulty oracle, a general lower bound on the number of oracle invocations necessary to find the marked item by \emph{any} quantum algorithm; ref.\ \cite{temme} proves the analogous result for any continuous-time implementation of oracle search. A concrete algorithm that functions under a faulty oracle is briefly suggested in \cite{SBW}, advocating the avoidance of active error correction as well, as we do.

In the quantum search part of this paper, we address the strengths and weaknesses of general oracular search algorithms under noise in more detail. Notably, we state all our upper and lower runtime bounds with explicit prefactors as e.g.\ in Theorems \ref{thmsymmalgomemoryless}--\ref{thmlowerboundwithmemory}, which in particular allows a comparison with the performance of classical search algorithms already for ``small'' (i.e.\ non-asymptotic) database sizes $N$. On the one hand we give fault-ignorant algorithms, on the other we investigate their optimality. More detailed comparisons to the works \cite{RegevSchiff,SBW,temme} are made in Section \ref{moregeneralmemoryalgos}. The Hilbert space in our setup consists of a search space, possibly supplemented by an ancillary quantum system and both suffering from noise; our constructive algorithms will actually not use ancillary systems, but we allow for them in the proofs of our lower runtime bounds, which strengthens these. Additionally, the algorithms may have access to a noiseless classical memory, which is a technologically realistic assumption. The noise itself is modelled as discrete hits of some noise channel that is to be applied between any two oracle invocations.

The problem of fault-ignorant quantum search is then as follows:\ Devise a quantum algorithm that, except for some specified maximal failure probability $\varepsilon$, returns the marked item under any decoherence rate $p$, using as few oracle calls as possible. The fastest (noiseless) classical algorithm needs $\lceil(1-\varepsilon)N\rceil$ table lookups for this task, examining the database entries one by one. We exhibit quantum algorithms which, under low but constant depolarizing level $p$, require fewer oracle calls than this classical algorithm, see e.g.\ Theorem \ref{thmunknownpwithoutexcluding}.

\bigskip

The paper is organized as follows. We give a more detailed description of the fault-ignorant idea in Sec.\ \ref{moreformaldescription}, while referring to Appendix \ref{sec:faultignorance} for a general and precise mathematical definition of \emph{fault-ignorant algorithms}. Those readers interested mainly in the problem of quantum search on a noisy quantum computer are directed to the remaining sections. In Sec.~\ref{sec:memoryless} we introduce quantum search in the presence of decoherence, and develop a fault-ignorant algorithm (Algorithm \ref{algounknownpnotexcluding}) that consists of ``quantum building blocks''  and uses the noiseless classical memory merely to store the marked item in case of a previous successful round. A matching lower bound on the runtime of such algorithms is given by Theorem \ref{thmsymmalgomemoryless}. We expand this analysis in Sec.~\ref{sec:withmemory} and allow for a noiseless classical memory that can store all previously falsified items, which enhances search efficiency (Algorithm \ref{algounknownpexcluding}). In Appendix \ref{sec:noisemodels} we discuss in more detail the noise models for which our results from the main text apply.

\subsection{Fault-ignorant algorithms -- definition and basic properties}\label{moreformaldescription}
Here we describe the main desiderata on fault-ignorant algorithms, with particular emphasis on oracular algorithms that are the main topic of the following sections. For a precise mathematical formalization of fault-ignorance, which also includes algorithms computing probabilistic functions as in sampling problems and computational algorithms such as factoring, we refer the reader to Appendix \ref{sec:faultignorance} (in particular, Definitions \ref{definitionnoisytask} and \ref{definefaultignorantalgo} there).

The tasks we consider consist in the computation of a classical output $o\in O$ as a function of an oracle index $x\in X$ and an input $i\in I$, which we assume is given via a specified encoding $\rho_i$ in the quantum register at the start of the algorithm. The specification of the task contains already the available size of the quantum register, i.e.\ the Hilbert space $\mathcal{H}$ on which the $\rho_i$ act and which is assumed to be fixed throughout the computation; in early implementations of quantum computers this size may be severely limited and is thus assumed to be part of the problem specification. Besides the quantum register, we allow for a classical register that may serve several purposes:\ {\it{(i)}} to be used during the computation, {\it{(ii)}} to store the output, {\it{(iii)}} to hold a binary flag indicating whether the output has already been written into the register, such that it can be read out by an outside agent without disturbing the (quantum) computation that may still be ongoing (such an indication is necessary since the noise level and the algorithm runtime may not be known in advance, see below).

For the task of oracle search among $N$ items, we thus have $x\in\{1,\ldots,N\}$ and would like to produce the output $o=x$, whereas the quantum register can be initialized in any fixed state $\rho_0$ as there is no other input to the task, i.e.\ $I=\{0\}$. We can allow for any quantum register $B({\mathcal{H}})$; ${\mathcal{H}}=\mathbb{C}^N$ would for example enable to perform Grover's algorithm using the oracle Eq.\ (\ref{oracleaction}) (see Section \ref{setupnoiselessandmodels}), but we may also allow for additional quantum registers such that e.g.\ ${\mathcal{H}}=\mathbb{C}^N\otimes\mathbb{C}^M$.

We have not yet specified in what way fault-ignorant algorithms should behave with respect to noise. For this we need two more ingredients. The first is a family of noise channels, denoted by $D_p$ and acting only on the quantum register, that should model the effect of the noise per unit time (see also Appendix \ref{sec:noisemodels}), and we think of the index $p\in[0,1]$ as a noise strength parameter that is \emph{not} known to the agent executing the algorithm. Classical registers are assumed noise-free. For our results on noisy quantum seach in the following sections, we will for example choose the noise models in Eqs.\ (\ref{partiallydepolnoise}) or (\ref{mostgeneralnoisemodel}). The second ingredient is the specification of the set $S$ of operations that an algorithm can perform per unit time. $S$ may be any subset of all quantum channels acting jointly on the quantum and classical register. Actually, for oracular algorithms, each element $T\in S$ also depends on the oracle index $x$ that is not directly accessible to the algorithm; rather, when an algorithm ``calls'' the operation $T$, then the quantum channel $T(x)$ is executed. Again for the oracular search case below, we allow any
\begin{equation}
T(x)~=~C_2\circ O_x\circ C_1\qquad(\text{for all}~x\in X)
\end{equation}
that calls the oracle $O_x$ from Eq.\ (\ref{oracleaction}) once, possibly preceded and followed by any quantum channels $C_1,C_2$ acting on the quantum and classical registers (and not depending on $x$). For oracular problems, such a specification of the allowed operations per unit time is quite natural, but other choices are possible; for computational problems one can for example impose locality restrictions on the operations per time step, such as in Example \ref{binaryadditionexample} in Appendix \ref{sec:faultignorance}. The specification of what constitutes one time step will be implicit in Eq.\ (\ref{algorithmicsequenceIA}) below, saying that the noise channel $D_p$ is to be applied once between any two operations from $S$.

An \emph{algorithm} is now simply a sequence $(T_n)_{n\in\mathbb{N}}$ of operations $T_n\in S$. Algorithms may depend on an accuracy parameter $\varepsilon\in(0,1)$, denoting the maximum probability with which we allow a wrong answer to be output.

Crucially however, for a \emph{fault-ignorant algorithm} $(T_n)_n$, we \emph{do not} allow the operations $T_n$ to depend on the above noise level $p$, but still require that, for \emph{each} noise level $p\in[0,1]$, after the execution of any number $t\geq t(p)$ of time steps
\begin{equation}\label{algorithmicsequenceIA}
T_t(x)D_pT_{t-1}(x)D_p\ldots D_pT_2(x)D_pT_1(x)\left(\rho_i\otimes|0\rangle_{class}\langle0|\right)
\end{equation}
the classical output register holds the correct output $o$, up to failure probability $\varepsilon$ (see also Fig.\ \ref{fig:ignorance} in Appendix \ref{sec:faultignorance} for illustration). We call an algorithm $(T_n)_n$ \emph{fault-ignorant} if $t(p)$ can be chosen finite for all $p\in[0,1]$, and we will actually denote by $t(p)$ the smallest time such that the previous requirement is satisfied for all $t\geq t(p)$ (see Definition \ref{definefaultignorantalgo} in Appendix \ref{sec:faultignorance} for an exact statement). In other words, a fault-ignorant algorithm is ignorant of the noise level $p$, but should nevertheless output the correct answer $o$ irrespective of the noise $p$, within a time $t(p)<\infty$ that may depend on the unknown noise level $p$.

Note that any computation which can be performed in a noiseless classical register of sufficient size alone constitutes a fault-ignorant algorithm, as we assume classical memory to be unaffected by the noise.

As a further instructive example of a fault-ignorant algorithm, assume the following:\ {\it{(a)}} $(T_n)_{n=1}^k$ is a finite-step algorithm such that, after executing the sequence (\ref{algorithmicsequenceIA}) for $k$ time steps, there is for \emph{any} noise level $p\in[0,1]$ a nonzero success probability $p_s(p)>0$ of obtaining the correct output $o$ in the classical register, {\it{(b)}} the set $S$ contains operations that allow to check whether a tentative result $o'$ in the classical register is correct, {\it{(c)}} the input encoding $\rho_i$ is ``classical'' in the sense that there are operations in $S$ which can extract from $\rho_i$ the index $i$ onto a classical register, and conversely allow to prepare the quantum state $\rho_i$ given the classical value of $i$ (requirement {\it{(c)}} is easy when $|I|=1$, as in the search task). If these three conditions hold, one can construct a fault-ignorant algorithm that solves the task for any fixed accuracy parameter $\varepsilon>0$, as follows:
\begin{itemize}
\item repeatedly run the $k$-step sequence $(T_n)_{n=1}^k$ with noise interspersed as in (\ref{algorithmicsequenceIA}) on the initial state $\rho_i\otimes|0\rangle_{class}\langle0|$ (note that, after the value $i$ has been saved in a classical register at the start, this initial state can be reproduced before each iteration due to condition {\it{(c)}}),
\item after each iteration check for the correctness of the answer proposed in this iteration according to {\it{(b)}}, and store the correct result in an overall output register once it has been found.
\end{itemize}
If the actual noise level is $p$, then after $r$ iterations of $(T_n)_{n=1}^k$, the failure probability equals $(1-p_s(p))^r$, which will become smaller than any specified $\varepsilon\in(0,1)$ if only $r\geq r(p)$ is large enough. More precisely, since each round including verification consumes $k+1$ time steps, the composite check-and-repeat algorithm has a runtime
\begin{equation}\label{genericruntimefaultignorant}
t(p)~=~(k+1)\left\lceil\frac{\log\frac{1}{\varepsilon}}{\log\frac{1}{1-p_s(p)}}\right\rceil~.
\end{equation}

For example, in searches using the quantum oracle (\ref{oracleaction}), one can satisfy the above conditions {\it{(a)}}--{\it{(c)}} by randomly selecting in each iteration an index $o'\in\{1,\ldots,N\}$ (see also the beginning of Section \ref{firstsearchwithexclusionsubsection}) and then using the oracle $O_x$ once to check whether $o'=x$. In this example, $k=0$ and $p_s(p)=1/N$ for all $p$. For another similar example, see the beginning of Section \ref{subsectionmemorylessunknown}. In fact, all the constructive algorithms we present in this paper will be variations of the above check-and-repeat algorithm, with possibly varying numbers $k_g$ of oracle uses in each round $g=1,\ldots,r$ (Section \ref{subsectionmemorylessunknown}, see Fig.\ \ref{fig:alg1} for illustration) and possibly leaving out previously falsified items in future rounds (Section \ref{sec:withmemory}, see also Fig.\ \ref{fig:alg2}).

Due to this simple way of constructing fault-ignorant algorithms, the meaningful question, which we investigate in the following sections for noisy quantum search, is about the efficiency of fault-ignorant algorithms. Notice for example that the first algorithm proposed in the previous paragraph has a runtime of $t(p)\approx N\log\frac{1}{\varepsilon}$, proportional to the problem size $N$ and \emph{independent} of $p\in[0,1]$. However, one might hope that there exist fault-ignorant algorithms which for small actual noise level $p$ need fewer oracle calls, because at least when the noise is \emph{known} to vanish (i.e.\ $p=0$) then Grover's algorithm solves the problem with at most roughly $\frac{\pi}{4}\sqrt{N}\log\frac{1}{\varepsilon}$ oracle calls (see Section \ref{setupnoiselessandmodels}). Furthermore, both these runtimes diverge in the limit of perfect accuracy, i.e.\ $\varepsilon\to0$, whereas the classical algorithm checking the $N$ items one after another needs only a finite number $N$ of steps for perfect accuracy. In fact, in the following sections we will develop a fault-ignorant quantum search algorithm which, under the noise models (\ref{mostgeneralnoisemodel}) and up to a constant overall factor, for any $N$, $p$, $\varepsilon$ requires fewer oracle calls than the ones just mentioned (see Theorem \ref{thmknownpwithoutexcluding}).

When given a fault-ignorant algorithm solving one specific task (e.g.\ one specified problem of size $N$, of specified quantum register size ${\rm dim}(\mathcal{H})$, noise model $N_p$, and accuracy goal $\varepsilon$), one can compare its runtime $t(p)$ to the runtime of other algorithms that can be implemented (e.g.\ to the classical search algorithm above, or to any algorithm that ``knows'' the noise level $p$ as one of its inputs, or to any algorithm that may use a quantum register of some larger size, etc.). On the basis of this comparison one can then decide whether to consider this fault-ignorant algorithm useful w.r.t.\ the competitor. Instead of solely the runtime, one may take into account also other factors in this comparison, such as the size of the quantum register used by either algorithm. It does not seem possible to give general criteria for such a decision. However, due to the prefactors given in the runtime bounds for our concrete algorithms, one can for example use our Theorems \ref{thmunknownpwithoutexcluding} and \ref{thmknownpwithoutexcluding} to compare these fault-ignorant algorithms in such a way to other algorithms (which we for example do below Theorem \ref{thmunknownpwithoutexcluding} and in Section \ref{moregeneralmemoryalgos}).

Instead of performing such a comparison for one specific task, one may consider a whole family of tasks (e.g.\ one for each problem size $N\in\mathbb{N}$ and accuracy $\varepsilon\in(0,1)$, possibly also allowing the size of the quantum register to vary independently of $N$, etc.) and fault-ignorant algorithms solving them. In this situation one can then investigate the \emph{scaling} of the runtime $t(p)$ with these parameters $N$, $\varepsilon$, etc., as is usual in complexity theory, and investigate various tradeoffs, e.g.\ between the runtime and the size of the quantum register. Again, the concrete questions seem to depend highly on the tasks at hand.

Nevertheless, since the main feature of fault-ignorant algorithms is to find the correct answer \emph{without knowing} the noise level $p$ in advance, we can introduce a distinguished notion of \emph{efficient fault-ignorant algorithm}: A family of fault-ignorant algorithms (each solving a task from a given family of tasks) is called \emph{efficient} if there exists a constant $C>0$ such that for each fixed $p\in[0,1]$ the runtime $t(p)$ (which depends on the member of the family) is at most a factor of $C$ larger than the runtime of any algorithm that solves the same task while \emph{knowing} the noise value $p$ as one of its inputs (i.e.\ need \emph{not} be fault-ignorant). In other words, we call a fault-ignorant algorithm \emph{efficient} if knowing the actual noise level $p$ would shorten its runtime by at most an overall factor $1/C$. Our Theorem \ref{thmlowerboundwithmemory} can thus be seen as a statement that Algorithm \ref{algounknownpexcluding} is efficient w.r.t.\ to the class of algorithms considered and, furthermore, an upper bound on the constant $C$ is apparent together with Theorem \ref{thmknownpwithoutexcluding}. Independently of this efficiency notion, for low enough (but unknown) noise level $p$, the runtime of Algorithm \ref{algounknownpexcluding} compares favorably with noiseless classical search (see Section \ref{moregeneralmemoryalgos}).

\section{Quantum search under noise --- memoryless approach}\label{sec:memoryless}

\subsection{Setup: Noiseless and noisy quantum search}\label{setupnoiselessandmodels}
The quantum search problem \cite{BBBV,nielsenchuang} asks for an algorithm of short runtime to identify (up to some small error probability $\varepsilon$) one out of $N$ oracles, i.e.~to return the index $x$ of the ``black box'' implementing the unitary transformation
\begin{equation}
\widetilde{O}_x:\mathbb{C}^N\otimes\mathbb{C}^2\to\mathbb{C}^N\otimes\mathbb{C}^2~,\quad\ket{x',b}\mapsto\ket{x',b\oplus\delta_{xx'}}~,\label{oraclebasic}
\end{equation}
where $x\in\{1,\ldots,N\}$ and $\oplus$ denotes addition modulo 2. Here, we assume the oracle $x$ to have been chosen uniformly at random (often referred to as unstructured database). It is customary to take the input of the oracle $\widetilde{O}_x$ at each step to be of the product form $\ket{\varphi}\otimes\frac{1}{\sqrt{2}}(\ket{0}-\ket{1})$ so that the output is also of this form with the \emph{sign} of the coefficient of $\ket{x}$ flipped in the first tensor factor $\ket{\varphi}$. In this case one can neglect the second subsystem and concentrate on the effective unitary transformation
\begin{equation}
O_x:\mathbb{C}^N\to\mathbb{C}^N~,\quad\ket{x'}\mapsto(-1)^{\delta_{xx'}}\ket{x'}~.\label{oracleaction}
\end{equation}
As usual, we measure the runtime of an oracular algorithm by counting the number of queries (oracle uses), which are viewed as the expensive or time-consuming operations, and disregard all the other quantum channels which are independent of the oracle.

Grover \cite{TheGrover} found a solution to this problem which makes use of the equal superposition state $\ket{\psi}\equiv\frac{1}{\sqrt{N}}\sum_{x=1}^N\ket{x}$ --- this state reflects the initial lack of knowledge about the oracle and will be used frequently in the following. At the beginning of the algorithm the state $\ket{\psi}$ is prepared and then the oracle black box $O_x$ and the unitary $I-2\ketbra{\psi}{\psi}$ are applied alternately. (In this description, and in our whole paper, we disregard any subsystem structure of $\mathbb{C}^N$; if, for example, $\mathbb{C}^N=(\mathbb{C}^2)^{\otimes n}$, then the ``inversion about the mean'' $I-2\ketbra{\psi}{\psi}$ can be implemented efficiently in the number $n$ of qubits \cite{nielsenchuang,TheGrover}). After $k$ applications of both operators, a von Neumann measurement is performed in the standard basis. The outcome of this measurement will give the correct index $x$ of the oracle with probability \cite{nielsenchuang}
\begin{equation}
\sin^2\left((2k+1)\arcsin{\frac{1}{\sqrt{N}}}\right)~,\label{noiselesssuccessprobab}
\end{equation}
independently of which oracle $x\in\{1,\ldots,N\}$ was implemented by the black box. In particular, if we choose $k=\lfloor\frac{\pi}{4}\sqrt{N}\rfloor$, the success probability is $1-O(N^{-1/2})$. Alternatively, if we fix some small maximal error probability $\varepsilon\gsim N^{-1/2}$, with which the algorithm may return an incorrect oracle index, then we may stop after $k=\lceil\frac{1}{4}\sqrt{N}\arccos(2\varepsilon-1)\rceil$ iterations. In fact, it can be shown, for any $0\le k<\frac{\pi}{4}\sqrt{N}$, that Grover's algorithm yields the highest probability of success which can be achieved by any quantum algorithm using the oracle $k$ times \cite{Zalka}.

\bigskip

The above analysis is valid when the unitaries and measurements etc.~can be implemented perfectly and the quantum computer is not subject to noise. In more realistic settings, however, these idealizations have to be lifted, and some such extensions have been considered in the literature before, cf.~Section \ref{sec:intro}. In this work, we consider the specific case where the quantum computer is continually affected by noise, e.g.~coming from the environment.

Our aim is not to approach this problem by implementing quantum error correction, which may be expensive in terms of the required control and size of the quantum computer. Rather, we aim to find (optimal) algorithms which succeed even under the influence of --- known, or ideally even \emph{unknown} --- noise, in such a way that their runtime may depend on the noise level; see Sec.\ \ref{sec:intro}.

\bigskip

Throughout this paper the term ``noise'' will mean the application of a certain quantum channel to the state of the quantum register in discrete time steps. This is supposed to model, within the discrete-time setting of oracle algorithms, that the quantum computer is continually affected by noise. More precisely, we will impose that the noise channel has to act \emph{once between any two invocations of the oracle}.

Our paradigmatic example of noise will be the family of \emph{partial depolarizing} channels
\begin{equation}
D_p(\varrho)~:=~p\frac{I_d}{d}\Tr(\varrho)\,+\,(1-p)\varrho~,\label{partiallydepolnoise}
\end{equation}
acting on states $\varrho$ on a $d$-dimensional Hilbert space. We can interpret these channels as acting on the system in a completely depolarizing way if a biased coin toss yields heads, which happens with probability $p\in[0,1]$, and otherwise leaving the quantum computer undisturbed. Intuitively speaking, the partially depolarizing noise (\ref{partiallydepolnoise}) discards the whole quantum register with probability $p$ between any two successive oracle invocations. In particular, quantum error correction cannot be applied to this noise model (cf.~Appendix \ref{sec:noisemodels}); but so it serves to illustrate our idea of \emph{fault-ignorant computing}, one of whose rationales actually is to avoid costly error correction procedures. In Appendix \ref{sec:noisemodels} we will further introduce partially dephasing noise (\ref{definepartialdephasing}), which has an additional interpretation as modelling the transitioning from quantum to classical algorithms, and we relate the different noise models.

The lower bounds on the runtime of noisy quantum search algorithms which we prove (Theorems \ref{thmsymmalgomemoryless} and \ref{thmlowerboundwithmemory}) rely on partial depolarizing (\ref{partiallydepolnoise}), which is a very drastic and in some implementations quite pessimistic noise model, as it acts in a strongly correlated way across the whole quantum register (somewhat similar to a noisy oracle \cite{RegevSchiff,SBW,temme}, see also Sec.\ \ref{moregeneralmemoryalgos}). For initial implementations of quantum computing this may in some cases indeed be a reasonable assumption, e.g.\ when $p$ denotes the occurrence probability of a noise event requiring the restarting of the whole quantum computer. The noise strength $p$ could for example be related to the timescale of a drifting laser or of collective hits by external stray magnetic fields. Nevertheless, it would be desirable to prove similar lower runtime bounds for weaker noise models, in particular incorporating some kind of locality, but for now our bounds provide at least a (pessimistic) point of comparison.

On the other hand, the concrete algorithms we provide will function with the guaranteed upper runtime bounds given in Theorems \ref{thmunknownpwithoutexcluding} and \ref{thmknownpwithoutexcluding} even under any more general noise of the form
\begin{equation}\label{mostgeneralnoisemodel}
D_p(\varrho):=pT(\varrho)+(1-p)\varrho~,
\end{equation} with $T$ being an arbitrary quantum channel. I.e.\ it is necessary only that with probability $(1-p)$ at each step no fatal noise event occurs. Partial depolarizing and partial dephasing are special cases of this.

\subsection{Building block for noisy search algorithms}\label{basicbuildingblocksubsection}
Let us see how the probability of a successful measurement, i.e.~returning the correct oracle index $x$, looks when we include an application of the noise channel after each query. As a preparation, we first consider Grover's algorithm under noise. Introducing the channel $G_x(\varrho)=\left((I-2\ketbra{\psi}{\psi})O_x\right)\varrho\left((I-2\ketbra{\psi}{\psi})O_x\right)^\dagger$, the final state can be written as $(D_p G_x)^k(\ketbra{\psi}{\psi})$, so that the success probability is then
\begin{equation}\label{firststatementsofpsinmaintext}
p_s(N,k,p)~:=~\sum_{x=1}^N\frac{1}{N}\bra{x}(D_p G_x)^k(\ketbra{\psi}{\psi})\ket{x}~,
\end{equation}
where we took an average over the $N$ possible oracles, since the search is unstructured and we assume equal a priori probabilities. In this paper, we choose to consider the \emph{average success probability} of algorithms, i.e.~averaged over all possible oracles with equal weight (see e.g.~\cite{Zalka}), as opposed to the minimal success probability, i.e.~minimized over all oracles (e.g.~\cite{BBHT,nielsenchuang}). Both figures of merit agree for ``symmetric algorithms'', e.g.~for Grover's algorithm \cite{Zalka} and for the constructive algorithms we propose in this paper. But our choice strengthens the lower bounds derived in the following on the required number of oracle invocations.

Now, for all above noise models $D_p=pD_1+(1-p)\id$ (see Eq.\ (\ref{mostgeneralnoisemodel})), the evolution $(D_p G_x)^k$ can be written as a sum of $2^k$ histories. Since each term gives a nonnegative contribution to the sum, we can find a lower bound by keeping only the noise-free term $(1-p)\id$ in each factor:
\begin{equation}
p_s(N,k,p)~\ge~(1-p)^k\sin^2\left((2k+1)\arcsin{\frac{1}{\sqrt{N}}}\right)~,\label{lowerboundprobs}
\end{equation}
which is quite sharp unless $kp\gg 1$, cf.~Appendix \ref{sec:noisemodels}; compare this also to noiseless case, Eq.~(\ref{noiselesssuccessprobab}). (The convention $0^0=1$ is understood throughout this paper.)

\bigskip

We are now interested in how well this simple algorithm, and other algorithms that we shall consider below (i.e.~not necessarily consisting of Grover steps), perform. That is, we would now like to derive \emph{upper} bounds on the success probability $p_s$ depending on $N$, $k$, and $p$. As the starting point we use the implicit bound on $p_s$ derived by Zalka \cite{Zalka} for the average success probability $p_s$ after $k$ oracle calls:
\begin{equation}
2N-2\sqrt{N}\sqrt{p_s}-2\sqrt{N}\sqrt{N-1}\sqrt{1-p_s}~\le~4k^2~.\label{zalkaimplicit}
\end{equation}
This bound has been established in \cite{Zalka} for the following situation:\ We start from any pure state $\ket{\phi}\in\mathbb{C}^K$; the oracle $O_x$ ($1\le x\le N$) inverts the coefficients of the basis states within a subset $S_x\subseteq\{1,\ldots,K\}$ where $S_x\cap S_y=\emptyset$ for $x\neq y$; we let arbitrary unitaries $U_1,\ldots,U_k\in\mathbb{C}^{K\times K}$ act after each oracle use; in the end we perform a von Neumann measurement in some basis, and our guess for $x$ is an arbitrary function of the measurement outcome.

The same bound (\ref{zalkaimplicit}) holds thus when we start from any mixed state over $\mathbb{C}^N\otimes\mathbb{C}^{M}$ ($M\geq1$), apply arbitrary channels between the oracle uses, and our guess for $x$ comes from measuring a POVM $(E_x)_{x=1}^N$. This holds because mixed states, quantum channels and POVM measurements can be dilated to pure states, unitary evolutions and von Neumann measurements on a larger system \cite{nielsenchuang}, and the oracles (\ref{oracleaction}) tensored by the identity on all other subsystems still invert coefficients of disjoint sets of basis states.

Even for the more general algorithms described above, we can thus use inequality (\ref{zalkaimplicit}) together with Lemma \ref{lemmaxyineq} (see Appendix \ref{app:technical}), to prove the following \emph{explicit upper bound} on the average success probability $p_s$ of any quantum algorithm using $k$ oracle calls (with or without noise):
\begin{equation}
p_s~\leq~\frac{(2k+1)^2}{N}~.\label{psboundfromzalka}
\end{equation}

This bound does not depend on the noise strength $p$, and thus gives no further restrictions on \emph{noisy} search compared to the noiseless case. But it does enable us to prove a result on the limitations of algorithms employing a noisy quantum register $\mathbb{C}^N\otimes\mathbb{C}^M$ for computation, with the oracle acting on the first subsystem. Note that the computational steps $T_i$, $T_i'$ in any algorithm covered by the following upper bounds on the success probability are nowhere assumed to be necessarily unital. If they all were unital, then after the occurrence of a noise hit of partial depolarizing (corresponding to the term $p\Tr(\varrho)I_d/d$ in (\ref{partiallydepolnoise})) the success probability of finding the marked item would be fixed at $1/N$, whereas non-unital actions might try to correct a partial depolarizing error and increase the success probability (we will comment on a particular non-unital error-detection-and-correction strategy below Eq.\ (\ref{eq:firstbound})). And, actually, the following result holds even for noise channels $D_p^\tau$ that may be somewhat more general than the partial depolarizing given in Eq.\ (\ref{partiallydepolnoise}):
\begin{theorem}[Bound on success probability of building block]\label{thm8kplus1overp}
Let $\varrho_0\in\B(\mathbb{C}^N\otimes\mathbb{C}^M)$ be the initial state of the algorithm, $\widehat{O}_x:\B(\mathbb{C}^N\otimes\mathbb{C}^M)\to\B(\mathbb{C}^N\otimes\mathbb{C}^M)$ defined by $\widehat{O}_x(|y\rangle\langle y'|\otimes\sigma)=(-1)^{\delta_{xy}+\delta_{xy'}}|y\rangle\langle y'|\otimes\sigma$ be the quantum channels implementing the oracles on the first subsystem, and $D_p^\tau(\varrho):=p\tau\Tr{\varrho}+(1-p)\varrho$ the noise channel that is to be applied between any two oracle calls, with $p\in[0,1]$ and $\tau\in\B(\mathbb{C}^N\otimes\mathbb{C}^M)$ any quantum state. Let $T_1,T_1',T_2,T_2',\ldots,T_k,T_k':\B(\mathbb{C}^N\otimes\mathbb{C}^M)\to\B(\mathbb{C}^N\otimes\mathbb{C}^M)$ denote steps in the algorithm, such that the state after $k$ uses of the oracle $\widehat{O}_x$ is
\begin{equation}
\varrho_k^x~=~T_k'\widehat{O}_xT_kD_p^\tau\ldots D_p^\tau T_2'\widehat{O}_xT_2 D_p^\tau T_1'\widehat{O}_xT_1 D_p^\tau(\varrho_0)~,\label{statefrombasicbuildingblock}
\end{equation}
and let the final measurement be given by the POVM $(E_y)_{y=1}^{N}$. Then the average success probability of this algorithm is upper bounded as follows:
\begin{equation}
p_s~:=~\sum_{x=1}^N\frac{1}{N}\Tr[\varrho_k^x E_x]~\leq~\frac{1}{N}+\frac{8}{Np^2}~,\label{moretrivialimplication}
\end{equation}
and
\begin{equation}
p_s~\le~8\frac{k+1}{Np}~.\label{8kplus1overp}
\end{equation}
\end{theorem}
\begin{proof}Introduce the following states:
\begin{equation}
\sigma_i^x~:=~T_k'\widehat{O}_xT_k T_{k-1}'\widehat{O}_xT_{k-1}\ldots T_{k-i+1}'\widehat{O}_xT_{k-i+1}(\tau)
\end{equation}
for $1\le i<k$, and $\sigma_k^x:=T_k'\widehat{O}_xT_k\ldots T_{1}'\widehat{O}_xT_{1}(\varrho_0)$.
With these we can write 
\begin{equation}
\varrho_k^x~=~\sum_{i=1}^{k}p(1-p)^{i-1}\sigma_i^x\,+\,(1-p)^k\sigma_k^x~,
\end{equation}
and hence
\begin{equation}
Np_s~=~\sum_x\Tr[\varrho_k^x E_x]~=~\sum_{i=1}^{k}p(1-p)^{i-1}\sum_x\Tr[\sigma_i^xE_x]\,+\,(1-p)^k\sum_x\Tr[\sigma_k^xE_x]~.
\end{equation}
As the ``computation'' of $\sigma_i^x$ involved $i$ oracle calls, from (\ref{psboundfromzalka}) we have $\sum_x\Tr[\sigma_i^xE_x]\le (2i+1)^2$, and thus
\begin{equation}\label{refereesays15}
\begin{split}
Np_s~
 & \le~\sum_{i=1}^{k}p(1-p)^{i-1} (2i+1)^2\,+\,(1-p)^k (2k+1)^2  \\
 & =~1+\frac{8}{p^2}\left(1-(1-p)^k\right)-\frac{8}{p}k(1-p)^k~,
\end{split}
\end{equation}
which trivially leads to (\ref{moretrivialimplication}). Furthermore, we can use $(1-p)^k\ge 1-kp$ by the Bernoulli inequality to obtain (\ref{8kplus1overp}):
\begin{equation}\label{refereesays16}
Np_s ~\leq~1+\frac{8}{p^2}\,kp~\le~8\frac{k+1}{p}~.
\end{equation}
\end{proof}

\begin{figure}
\begin{center}
\includegraphics{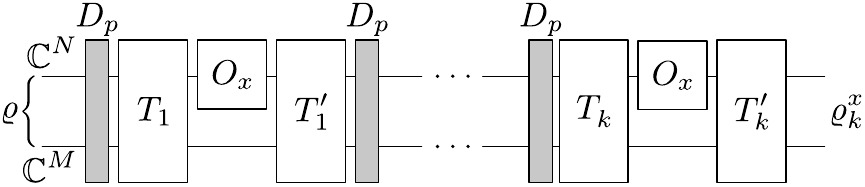}
\end{center}
\caption{\label{fig:alg0}Algorithm 0, performing $k$ steps of quantum search. Noise is acting between any two oracle invocations $O_x$ (see Eq.~(\ref{oracleaction})), before and after which one may however apply arbitrary channels $T_i$, $T'_i$. The computation uses a quantum register of dimension $N$, on which the oracle acts, and an ancillary (quantum) register $\mathbb{C}^M$, both noisy. A measurement, described by some POVM $(E_y)$, is applied to the final state $\varrho_k^x$ to guess the marked element $x$.}
\end{figure}

In the following text we shall address a sequence of operations as in (\ref{statefrombasicbuildingblock}) acting on an initial state $\varrho_0$ as \textbf{Algorithm 0} (or $\mathbf{Alg^0}$), which is also depicted in Fig.~\ref{fig:alg0}.

It is worthwhile to mention two ways of using bounds as in Theorem \ref{thm8kplus1overp}:
\begin{enumerate}
\item[\caseone] one can consider the noise level $p$ to be fixed and examine the scaling behaviour (e.g.\ of the algorithm runtime) with respect to the number $N$ of search items, or
\item[\casetwo] one can consider $p$ to scale in some way with $N$.
\end{enumerate}
Point \casetwo{} provides a way to compare the results with other works (e.g.~\cite{SBW}) where this kind of scaling was analyzed. Our results, however, apply to any values of $N$ and $p$ (and later, of $\varepsilon$) in the stated ranges, but we are implicitly often imagining the case $N\to\infty$, $p={\rm const}$ (and $\varepsilon={\rm const}$), which is a sensible limit as explained in Appendix \ref{sec:noisemodels}.

Theorem \ref{thm8kplus1overp} has two important implications. The first is that, for any fixed $p>0$, the growth of the success probability is at most \emph{linear} in the number $k$ of queries, as opposed to the quadratic growth in the noiseless case (cf.~(\ref{noiselesssuccessprobab}) for $k\ll\frac{\pi}{4}\sqrt{N}$). This may come as a surprise as one might have guessed that the quadratic speedup of Grover's algorithm may persist for small enough noise levels $p>0$ (i.e.~fixed $p$ and $N\to\infty$). In other words, Theorem \ref{thm8kplus1overp} says that there exists \emph{no} algorithm with success probability $p_s\sim k^2$ whenever partial depolarizing acts. The inequality (\ref{8kplus1overp}) in particular implies that quantum error correction \emph{cannot} be done for partially depolarizing noise $D_p$ with $p>0$.

The second implication is that the success probability is bounded by $\frac{1}{N}+\frac{8}{Np^2}$ \emph{independently of $k$}, cf.~(\ref{moretrivialimplication}). For growing $N\to\infty$, this goes to $0$ unless $p=O(N^{-1/2})$; in general we cannot reach a prescribed success probabilty $1-\varepsilon$ with an algorithm as described in Theorem \ref{thm8kplus1overp}. The straightforward solution to this problem is that we repeat the algorithm (including the final measurement) until the probability of failure in all the repetitions combined drops below $\varepsilon$. This strategy is detailed in the following subsection and shows the basic structure for all further algorithms.

\subsection{Search algorithm by repeating the basic building block --- known noise level $p$}\label{subsectrepeatbuildingblockknownp}

Now we consider algorithms that consist of a number of repetitions (rounds) of the basic building block described in the previous subsection:~repeatedly preparing states like (\ref{statefrombasicbuildingblock}), using the oracle $O_x$ a number of times, and trying to infer the oracle index $x$ by a measurement. In our concrete constructions we shall, as in Grover's algorithm \cite{TheGrover,nielsenchuang}, specify the channels $T'_i$ in (\ref{statefrombasicbuildingblock}) to perform unitary inversions ($I-2\ketbra{\psi}{\psi}$) about the mean, and $T_i$ to equal the identity channel (whereas for our lower bounds $T_i$, $T'_i$ remain arbitrary).

Note that these repetitions require a \emph{noiseless} classical memory in order to reliably store the correct result $x$ after one of the rounds has been successful (any noisy classical register may be part of the subsystem $\mathbb{C}^M$ in Theorem \ref{thm8kplus1overp} and is subject to noise $D_p$). Furthermore, in order to test whether the measurement after any one of the repetitions has returned the correct index $x$ and in order not to overwrite the correct result in subsequent rounds, one needs to perform a verification of the measurement outcome --- using \emph{one} oracle call or classical table look-up --- immediately after each measurement.

In this subsection we will first develop and analyze an algorithm that finds, except for some specified failure probability $\varepsilon>0$, the marked element $x$ among $N$ elements, on a noisy quantum computer. Secondly, we will prove (Theorem \ref{thmsymmalgomemoryless}) that, up to a constant, the runtime of this algorithm is optimal among a certain class of algorithms, when the noise level $p$ is known in advance.

Our algorithm will consist of $m$ rounds of the procedure described in Theorem \ref{thm8kplus1overp}, each time checking whether the concluding measurement gave the correct $x$ and, if so, storing the result in a noiseless classical register. Specifically, in each round we perform Grover's algorithm for some number $k$ (to be determined) of steps --- this has been described at the beginning of Subsection \ref{basicbuildingblocksubsection}. We now first give a motivating ``derivation'' of the algorithm.

Clearly, the events ``the noise did not hit and the measurement was unsuccessful in the $i$th round'' and ``the noise did not hit and the measurement was unsuccessful in the $j$th round'' are independent for $i\neq j$ when we use Grover iterations, and similarly for more than two rounds, since Grover's algorithm is symmetric with respect to permutation of the oracles (the probability for the complement of each such event is given by (\ref{lowerboundprobs}), which we will substitute below for $1-p_s$). This means that if we perform $m$ rounds with $k$ Grover iterations in each round, then the probability of failure and the total number of queries will be $(1-p_s(N,k,p))^m$ and $(k+1)m$, respectively (including the verification step after each of the $m$ measurements). If we are to reach the target error bound $\varepsilon$ in the least number of steps, we need to minimize $(k+1)m$ subject to the condition $(1-p_s(N,k,p))^m\le\varepsilon$. The latter gives a lower bound on $m$ depending on $N$, $k$ and $p$, namely
\begin{equation}
(k+1)m~\ge~\frac{k+1}{-\log(1-p_s(N,k,p))}\log\frac{1}{\varepsilon}~.\label{minnumbermemorylesseqn}
\end{equation}
This means that for given values of $N$ and $p$ one needs to minimize
\begin{equation}
R(N,k,p)~:=~\frac{k+1}{-N\log(1-p_s(N,k,p))}~,\label{definitionf}
\end{equation}
i.e.~the number of queries per decrease of error probability by a factor of $e$, where the factor $N^{-1}$ is included for normalization. Let the optimal $k$ be denoted by $k_{opt}(N,p)$. Note that minimizing the rate function $R$ from (\ref{definitionf}) originates from minimizing the number of oracle calls from independent rounds to get the failure probability below $\varepsilon$. This is different than optimizing the \emph{expected} number of oracle calls, cf.~\cite{BBHT,Zalka} for a comparison.

Then our new algorithm consists of $m=\left\lceil\left(\log\varepsilon\right)/{\log\left(1-p_s(N,k_{opt}(N,p),p)\right)}\right\rceil$ repetitions of rounds with $k_{opt}(N,p)$ Grover steps each, measurement in the standard basis plus verification step and storing the result in the classical output register when successful. The discussion in the following example will motivate an easy-to-compute quantity to be used instead of $k_{opt}(N,p)$.

\begin{example}[Asymptotically optimal number of Grover iterations]\label{examplewhereNequalsinfty}
If we are interested in large databases, we can simplify by taking $N\to\infty$, yielding, with (\ref{lowerboundprobs}) instead of (\ref{exactsuccessprobabfordepol}) (giving virtually identical results):
\begin{equation}\label{firsteqninexample}
R(\infty,k,p)~:=~\lim_{N\to\infty}R(N,k,p)~=~\frac{(k+1)}{(1-p)^{k}(2k+1)^2}~.
\end{equation}
To find the optimal $k=k_{opt}(\infty,p)$, we compare ($k\geq 1$):
\begin{equation}
\frac{R(\infty,k-1,p)}{R(\infty,k,p)}~=~\frac{k (2 k+1)^2}{(k+1) (2 k-1)^2}(1-p)~.
\end{equation}
This is a decreasing function of $p$, so $R(\infty,k-1,p)$ and $R(\infty,k,p)$ intersect at
\begin{equation}
p_k(\infty)~:=~\frac{4k^2+4k-1}{k (2 k+1)^2}~=~\frac{1}{k}+O(k^{-3})~,
\end{equation}
which is appropriate for large $k$, and we have $k(\infty,p)=k$ whenever $p_{k+1}(\infty)<p\le p_k(\infty)$ with the convention $p_0(\infty):=1$. By inverting the series expansion around $k=\infty$ one can get an \emph{explicit} approximate expression for $k_{opt}(\infty,p)$:~the inverse is $\frac{1}{p}+O(p)$, and thus $k_{opt}(\infty,p)\approx\lfloor\frac{1}{p}\rfloor$ is a good approximation especially for small $p$.

Intuitively, this behaviour of the optimal length $k_{opt}$ of a quantum round can be understood by noting that the quantum register remains undisturbed with reasonably high probability (of order $O(1)$) for time $\sim1/p$, whereas with a probability approaching $1$ exponentially the quantum register will be disturbed by noise when $k\gtrsim(const)/p$. This is because the noise will hit at each step independently with probablity $p$ (see Eq.\ (\ref{partiallydepolnoise})). Thus, roughly $\sim1/p$ Grover steps provide an advantage (see also, e.g., \cite{SBW}). Plugging $k=\lfloor\frac{1}{p}\rfloor=\frac{1}{p}+O(1)$ into $R(\infty,k,p)$ and considering small noise level $p\ll1$, we see
\begin{equation}\label{lasteqnexampole}
R\left(\infty,\frac{1}{p}+O(1),p\right)~=~\frac{e}{4}p+\frac{e}{4}p^2+O(p^3)~,
\end{equation}
so $R(\infty,k_{opt}(\infty,p),p)$ vanishes linearly in $p$ around $p=0$. The fact that this is non-zero and finite for $p>0$, means that the normalization by $1/N$ in (\ref{definitionf}) was the ``correct'' one. This suggests that the number of steps $m(k+1)$ necessary is proportional to $Np\log(1/\varepsilon)$, i.e.\ linear in $N$ for any fixed $p>0$ and $\varepsilon\in(0,1)$, meaning that the quadratic speedup is lost under depolarizing noise (\ref{partiallydepolnoise}); see Theorem \ref{thmsymmalgomemoryless} for a more rigorous and general lower bound on the runtime.
\end{example}

Example \ref{examplewhereNequalsinfty} motivates a more rigorous analysis of the case of finite $N<\infty$. In this finite case, one has to be careful for small values of $p$, since it is clearly not a good idea to do more than $\frac{\pi}{4}\sqrt{N}$ Grover iterations in one round. Let us first assume that $p>\frac{4}{\pi}\frac{1}{\sqrt{N}}$, and suppose we perform $m=\lceil cNp^2\log\varepsilon^{-1}\rceil$ rounds with $k=\lfloor\frac{1}{p}\rfloor$ Grover steps in each round, for some $c>0$ (this ansatz is motivated by Example \ref{examplewhereNequalsinfty}, and will below turn out to be good). Then the probability of failing in all rounds can be bounded as (see \ref{lowerboundprobs})
\begin{align}
\left[1-p_s\left(N,\left\lfloor\frac{1}{p}\right\rfloor,p\right)\right]^{\lceil cNp^2\log\frac{1}{\varepsilon}\rceil}~
 & \le~ \left[1-(1-p)^{\left\lfloor\frac{1}{p}\right\rfloor}\sin^2\left(\left(2\left\lfloor\frac{1}{p}\right\rfloor+1\right)\arcsin\frac{1}{\sqrt{N}}\right)\right]^{\lceil cNp^2\log\frac{1}{\varepsilon}\rceil} \nonumber \\
 & \le~ \left[1-(1-p)^{\frac{1}{p}}\sin^2\left(\left(\frac{2}{p}-1\right)\arcsin\frac{1}{\sqrt{N}}\right)\right]^{\lceil cNp^2\log\frac{1}{\varepsilon}\rceil}  \nonumber\\
 & \le~\left[1-(1-p)^{\frac{1}{p}}(1-\delta)\left(\frac{2}{p}-1\right)^2\frac{1}{N}\right]^{\lceil cNp^2\log\frac{1}{\varepsilon}\rceil} \nonumber \\
 & \le~\exp\left\{-c(1-p)^\frac{1}{p}(1-\delta)\left(\frac{2}{p}-1\right)^2p^2\log\frac{1}{\varepsilon}\right\}\nonumber\\
&=~ \varepsilon^{c(1-\delta)(1-p)^{\frac{1}{p}}(2-p)^2}\label{boundfailingforknownp}
\end{align}
for some $0<\delta\leq 1-4/\pi^2$ depending on $N$ and $p$.

As we want to guarantee a failing probability of at most $\varepsilon$, we need to choose $c$ such that the exponent in the final expression (\ref{boundfailingforknownp}) is at least $1$ --- independently of $p$ for the following statements to be valid. However, for large values of $p$ the exponent goes to $0$, which is a consequence of vanishing (\ref{lowerboundprobs}) for large $p$ and $k=\lfloor\frac{1}{p}\rfloor\ge 1$ Grover steps; this can be avoided by introducing a cutoff $p^*\in(0,1)$ into the specification of the algorithm, such that for $p\geq p^*$ we use $k=0$ iterations in each round, i.e.~only perform verification steps on randomly chosen elements. The failure probability in this range is
\begin{equation}
\left(1-\frac{1}{N}\right)^{\lceil cNp^2\log\frac{1}{\varepsilon}\rceil}~\le~\left(1-\frac{1}{N}\right)^{cN{p^*}^2\log\frac{1}{\varepsilon}}~\le~\varepsilon^{c{p^*}^2}~.\label{failureproblargep}
\end{equation}

Numerically one finds that viable values for $c$ and $p^*$ in the specification of a concrete algorithm, i.e.~such that both (\ref{boundfailingforknownp}) and (\ref{failureproblargep}) do not exceed $\varepsilon$, are, for example, $c=5$ and $p^*=1/2$ (even when setting $\delta=1-4/\pi^2$). We have not optimized these constants, as our main interest for now is in the scaling for large $N$, small $\varepsilon$, and all noise levels $p$. The number of oracle invocations during such an algorithm is upper bounded by
\begin{equation}
\left(\left\lfloor\frac{1}{p}\right\rfloor+1\right)\left\lceil cNp^2\log\frac{1}{\varepsilon}\right\rceil~
 \le~\frac{2}{p}\left(cNp^2\log\frac{1}{\varepsilon}+1\right)   ~\le~ \frac{\pi}{2}\sqrt{N}+2cNp\log\frac{1}{\varepsilon}~,\label{upperboundnumsteps}
\end{equation}
where we used $p>\frac{4}{\pi}\frac{1}{\sqrt{N}}$.

For small noise levels $p=\frac{\beta}{\sqrt{N}}$ (where $0<\beta\le4/\pi$), one can do $\lfloor\frac{\pi}{4}\sqrt{N}\rfloor$ Grover iterations in each of the $m$ rounds, i.e.~before each measurement, which gives a success probability of at least roughly $e^{-\beta\pi/4}$ by the lower bound (\ref{lowerboundprobs}); therefore, $m=\lceil e^{\beta\pi/4}\log\frac{1}{\varepsilon}\rceil\le\lceil e\log\frac{1}{\varepsilon}\rceil$ rounds are sufficient to get the failure probability below $\varepsilon$, putting an upper bound on the number of oracle calls similar to (\ref{upperboundnumsteps}) with a term proportional to $\sim\sqrt{N}\log(1/\varepsilon)$ instead of the term $\sim Np\log(1/\varepsilon)$.

Summing up, our algorithm finds the marked element $x$, except with probability $\varepsilon\in(0,1]$, on a quantum computer with noise level $p\in[0,1]$ using the oracle at most
\begin{equation}\label{eq:firstbound}
c_1\sqrt{N}+c_2(Np+\sqrt{N})\log\frac{1}{\varepsilon}
\end{equation}
times, e.g.~for (non-optimal) constants $c_1=2,c_2=10$. We omit here a formal statement of the algorithm, which should have become resonably clear from the description above, but will remedy this in Subsection \ref{subsectionmemorylessunknown} for the more general case of unknown noise level.

The algorithm just described performs a number $m$ of quantum rounds, each of identical length $k$ given by an ansatz that is based on Example \ref{examplewhereNequalsinfty}. A cleverer algorithm might try to detect when a noise event has happened and then immediately abort the present round in such a case and start a fresh round in order not to ``waste'' oracle uses. One way to accomplish this would be to adjoin to the quantum register $\mathbb{C}^N$ used above another quantum register $(\mathbb{C}^2)^{\otimes r}$ of $r$ qubits that is initialized to $\ket{0}^{\otimes r}$ at the beginning of each round and is left untouched by the above algorithm. In case of a partial depolarizing noise event, given by the term $pI_{2^rN}(\Tr{\varrho})/(2^rN)$ in (\ref{partiallydepolnoise}), the $r$-qubit register will then be reset to a computational basis state other than $\ket{0}^{\otimes r}$ which can be detected by a projective measurement on this auxilliary system with probability $1-2^{-r}$. Thus, by expending a small number $r$ of extra qubits (e.g.\ a number $r$ that is constant in the problem size $N$, or chosen as $r\sim\log(1/\varepsilon)$) one can detect a noise hit with high probability and abort the present round to gain a saving in the number of oracle calls compared to the algorithm outlined above.

While this is a viable strategy in the noise model used above, it is actually extremely dependent on the noise model (\ref{partiallydepolnoise}). If, for example, the noise would replace the whole quantum state with probability $p$ by $\ket{0_N}\otimes\ket{0}^{\otimes r}$ (instead of $I_{2^rN}/(2^rN)$), then the exact strategy would not work anymore. In particular, any such strategy relies strongly on the fact that the noise is correlated across the whole quantum register. While we allow such strongly correlated noise as a pessimistic assumption from the outset, in particular to prove our \emph{lower} runtime bounds, one would probably not want the actual \emph{constructive} algorithms to rely on this assumption. In contrast, our algorithm developed below Example \ref{examplewhereNequalsinfty} as well as the upper runtime bound (\ref{eq:firstbound}) work for any noise model $D_p^{T}(\varrho)=pT(\varrho)+(1-p)\varrho$ with $T$ \emph{any} quantum channel $T$ (see App.\ \ref{sec:noisemodels}). This is because we only use Eq.\ (\ref{lowerboundprobs}), which merely relies on the fact that with probability $1-p$ the quantum register remains undisturbed. Furthermore, even when relying on an exactly known noise model as e.g.\ in Eq.\ (\ref{partiallydepolnoise}), one would at most save a constant factor of order $1$ by the error-detection-strategy compared to the runtime (\ref{eq:firstbound}) of the algorithm outlined above. This is due to the exponential first factor in (\ref{lowerboundprobs}), which implies that only with small probability $\sim1-e^{-C}$ will the noise hit occur before executing $C/p$ steps in one round (where $C<1$ here, such that there would be a saving).

We would like to point out that there are at least two different interpretations of runtime complexity results like Eq.\ (\ref{eq:firstbound}). Firstly, one can run the algorithm indefinitely long (i.e.~without any a priori bound on the number of rounds) until the marked element is found. Then we can guarantee that the algorithm gives the correct result with probability $1$, and the number of oracle calls required is at most $c_1\sqrt{N}+c_2(Np+\sqrt{N})\log\frac{1}{\varepsilon}$ except with probability $\varepsilon$. Alternatively, one can decide in advance to use the oracle $c_1\sqrt{N}+c_2(Np+\sqrt{N})\log\frac{1}{\varepsilon}$ times before terminating the above algorithm, and after any successful measurement store the result in a classical memory; then, in the end, the marked element will have been found with probability at least $1-\varepsilon$.

With the runtime bound (\ref{eq:firstbound}) at hand, one can look at the case where $p$ is fixed and independent of $N$. Then we see that, unless $p=0$, the leading term is $c_2Np\log(1/\varepsilon)$, i.e.\ proportional to $N$. On the other hand, if one supposes that $p$ scales decreasingly with $N$, the other terms may dominate. In particular, if $p\lsim1/\sqrt{N}$, the leading term is $c_2\sqrt{N}\log(1/\varepsilon)$.

\bigskip

Next we show that our algorithm presented above is essentially optimal within a certain class of algorithms (a wider class of algorithms will be considered in Section \ref{sec:withmemory}). Namely, we assume that the algorithms employ a quantum register $\mathbb{C}^N\otimes\mathbb{C}^M$ (as above in Theorem \ref{thm8kplus1overp}), consist of several ``rounds'' where in each round we prepare some state, apply arbitrary channels and use the oracle an arbitrary number of times (possibly different for each round, but applying the noise channel between any two consecutive queries), do any measurement yielding an element of $\{1,\ldots,N\}$, and verify the result with one oracle use, writing it into a (noiseless) classical register reserved for storing the output if correct. Crucially, we assume that the events of success in each round are independent of each other. This assumption is valid in particular if the noise is symmetric (under permutations of the basis vectors $\ket{x}$ of $\mathbb{C}^N$, which partial depolarizing from Eq.\ (\ref{partiallydepolnoise}) satisfies) and if the steps between measurements are Grover iterations (as for example in our algorithm above).
\begin{theorem}[Lower runtime bound on memoryless algorithms]\label{thmsymmalgomemoryless}
Consider a sequence of algorithms, one for each size of the search space $N=1,2,3,\ldots$, satisfying the assumptions just stated and subject to partial depolarizing noise (Eq.\ (\ref{partiallydepolnoise})). If the success probabilities are $1-\varepsilon_N$, then, asymptotically, the number of queries $q_N$ in the $N$th algorithm is lower bounded by the level $p\in[0,1]$ of depolarizing noise:
\begin{equation}
\label{eq:asymptotics}
\liminf_{N\to\infty}\frac{q_N}{N\log(1/\varepsilon_N)}~\ge~\frac{p}{8}~.
\end{equation}
More precisely, for any $p,\varepsilon\in(0,1]$ and any finite $N>9/p^2$, the number of queries $q_N$ satisfies:
\begin{equation}\label{finiteNbound}
q_N~\ge~\frac{Np\log(1/\varepsilon)}{8}\,\left[-\frac{Np^2}{9}\log\left(1-\frac{9}{Np^2}\right)\right]^{-1}~.
\end{equation}
\end{theorem}
\begin{proof}
Suppose that the $N$th algorithm consists of $m_N$ rounds with $k^1_N,\ldots,k^m_N$ queries in each round (abbreviating $m\equiv m_N$), with failure probability $\varepsilon^i_N$ in the $i$th round. Then
$q_N=(k^1_N+1)+\ldots+(k^m_N+1)$, and by the independence condition
\begin{equation}\label{weightedaverageinproof}
\frac{q_N}{N\log(1/\varepsilon_N)}~=~\frac{(k^1_N+1)+\ldots+(k^m_N+1)}{N\log\left(1/(\varepsilon^1_N\ldots\varepsilon^m_N)\right)}~=~\frac{\frac{k^1_N+1}{N\log(1/\varepsilon^1_N)}\log(1/\varepsilon^1_N)+\ldots+\frac{k^m_N+1}{N\log(1/\varepsilon^m_N)}\log(1/\varepsilon^m_N)}{\log(1/\varepsilon^1_N)+\ldots+\log(1/\varepsilon^m_N)}~,
\end{equation}
which is a weighted average of expressions of type $\frac{k+1}{N\log(1/\varepsilon(N,k,p))}$. Lower-bounding this expression thus automatically lower-bounds (\ref{weightedaverageinproof}), and therefore it is enough to consider the $m=1$ case only.

Since, by (\ref{moretrivialimplication}), $1-\varepsilon(N,k,p)\le\frac{1}{N}+\frac{8}{Np^2}\to 0$ as $N\to\infty$ (for any $p>0$), one has $\log\left(1/\varepsilon(N,k,p)\right)\le\delta_N(1-\varepsilon(N,k,p))$ for some positive sequence $\delta_N\to1$. Using that, by (\ref{8kplus1overp}), also $1-\varepsilon(N,k,p)\le8(k+1)/(Np)$ we get 
\begin{equation}
\liminf_{N\to\infty}\frac{k+1}{N\log\left(1/\varepsilon(N,k,p)\right)}~ \ge~\liminf_{N\to\infty}\frac{k+1}{N\delta_N\cdot8(k+1)/(Np)}~=~\frac{p}{8}~.
\end{equation}

The finite-$N$ bound follows from $1-\varepsilon(N,k,p)\leq\frac{9}{Np^2}$ and setting $\delta_N$ equal to the quantity inside the square brackets in (\ref{finiteNbound}).
\end{proof}

Theorem \ref{thmsymmalgomemoryless} shows that, under any nonzero noise $p>0$ (and $\varepsilon\in(0,1)$), our algorithm from above has asymptotically optimal runtime, up to a constant factor:~In our algorithm, $\varepsilon_N\equiv\varepsilon$ was chosen independent of $N$ and (\ref{eq:firstbound}) shows that asymptotically $q_N\lsim10Np\log(1/\varepsilon)$, which matches (\ref{eq:asymptotics}) up to a factor of $80$. In the noiseless case $p=0$, our algorithm reduces to repeated Grover searches, whose optimality for $p=0$ was shown in \cite{BBBV,BBHT,Zalka}.

One other implication is noteworthy: On the one hand, Theorem \ref{thmsymmalgomemoryless} says that, at fixed positive noise $p>0$ and asymptotically for $N\to\infty$, the number of oracle queries $\gsim Np\log(1/\varepsilon)$ \emph{has} to grow at least linearly in $N$, so that the quadratic speed of noiseless Grover search is lost (at least for the class of algorithms considered above, and for the depolarizing noise model, Eq.\ \ref{partiallydepolnoise}). On the other hand, however, the prefactor in this linear growth is $O(p)$, which is actually achieved by the explicit algorithm above, see Eq.\ (\ref{eq:firstbound}); thus, for small enough noise $p$, the number of oracle calls to solve the search task by a quantum algorithm is \emph{much less} than the minimal number $\sim N(1-\varepsilon)$ of oracle calls required by any classical algorithm, even in a noiseless environment.

 In the following subsection, we will extend the above algorithm so that it works in a noisy environment even when its noise level $p$ is \emph{not} known beforehand (Algorithm \ref{algounknownpnotexcluding} and Theorem \ref{thmunknownpwithoutexcluding}).

\subsection{Fault-ignorant search composed from basic building blocks}\label{subsectionmemorylessunknown}
We are now ready to turn to the \emph{``fault-ignorant'' setting} --- the algorithm should be ignorant of the actual noise level under which the quantum computer operates. More precisely, the goal is to find an algorithm for which not the ability to give the correct result depends on the level of noise, but rather only its runtime may depend on the actual noise level. Actually, the algorithm described in the previous subsection has this property for any fixed number $k$ of oracle calls per round; however, the runtime can then become large unless $k\approx k_{opt}(N,p)$. For example, if we choose $k\approx\frac{\pi}{4}\sqrt{N}$ in order to get a quadratic speedup for $p=0$, then for $p\approx 1$ the number of oracle calls grows as fast as $N^{3/2}$, which is clearly unsatisfactory.

In order to overcome this problem we allow the number of queries to change from round to round. Thus, for each $N$ and $\varepsilon$, we need to find a sequence $k_1(N,\varepsilon),\,k_2(N,\varepsilon),\,\ldots$, where $k_i(N,\varepsilon)$ denotes the number of Grover iterations performed in the $i$th round. Again, for our constructive algorithm, we employ the usual Grover iterations; and again, Theorem \ref{thmsymmalgomemoryless} will later show that this algorithm is nearly optimal.

One idea can be as follows. In the first round, we do a Grover search with $k_1(N,\varepsilon):=k_{opt}(N,0)\approx\frac{\pi}{4}\sqrt{N}$ oracle calls (for the definition of $k_{opt}$ see below (\ref{definitionf})). For $N\gsim\varepsilon^{-2}$ this is enough to get the error probability below $\varepsilon$ as long as $p=0$; the set $\left\{p\in[0,1]\,\big|\,p_s(N,k_{opt}(N,0),p)>1-\varepsilon\right\}$ is open and therefore 
\begin{equation}
p_2~:=~\inf\left\{p\in[0,1]\,\big|\,1-p_s(N,k_1,p)\geq\varepsilon\right\}~
\end{equation}
exists and is larger than $0$ (if the set is empty, e.g.~when $\varepsilon\ll N^{-1/2}$, we set the infimum to $p_2:=0$).
Suppose that the measurement after the first round fails to find the marked element $x$. There are now two possibilities:~either the actual noise level was below $p_2$, in which case the probability of this failure was at most $\varepsilon$ (i.e.~as required); or the actual noise level exceeded $p_2$, in which case the function $k_{opt}$ gives an upper bound on the optimal number of Grover iterations to perform in the next round, so we set $k_2(N,\varepsilon):=k_{opt}(N,p_2)$. We then proceed similarly by iteratively setting $p_i:=\inf\left\{p\in[0,1]\,\big|\,\prod_{j=1}^{i-1}\left(1-p_s(N,k_j,p)\right)\ge\varepsilon\right\}$ and $k_i:=k_{opt}(N,p_i)$, giving the number of Grover iterations to be performed in the $i$th round.

\bigskip

\begin{figure}
\begin{center}
\includegraphics{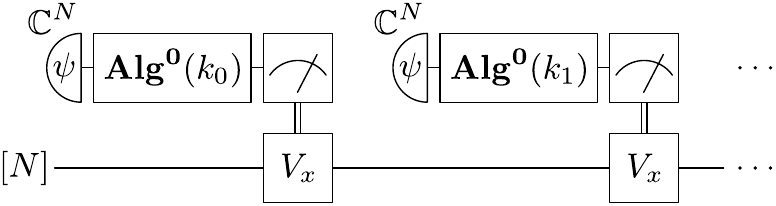}
\end{center}
\caption{\label{fig:alg1}The fault-ignorant Algorithm \ref{algounknownpnotexcluding} searches the marked element $x$ in consecutive rounds of $k_0,k_1,k_2,\ldots$ Grover steps plus one verification step each. Each round starts by preparing the equal superposition state $\psi$, which is follwed by Grover steps (a special case of Algorithm $0$, see Fig.~\ref{fig:alg0}), and a measurement in the computational basis; finally, the output is checked against the oracle (possibly by list look-up) and, in case of success, stored (by $V_x$) in a noiseless classical register with $N$ states, so that the result is ready for readout by an external agent (the algorithm may however continue). No ancillary system $\mathbb{C}^M$ is used by the algorithm ($M=1$).}
\end{figure}

The sequences $\{k_i\}_i$ obtained this way are difficult to analyze, but by examining the results of numerical simulations for various values of $N$ and $\varepsilon$ one can get an idea about their behaviour. This turns out to be enough to get an approximation which still achieves the asymptotically optimal performance, up to a multiplicative factor in the runtime (see below). Specifically, we arrive at the following algorithm (see also Fig.~\ref{fig:alg1}) for \emph{fault-ignorant quantum search}:
\begin{algorithm}[Quantum search from basic building blocks]\label{algounknownpnotexcluding}
For suitably chosen $c>0$, define
\begin{equation}\label{alphapresciptalg1}
\alpha_g(\varepsilon)~:=~\frac{1}{\sqrt{1+\frac{g}{c\log(1/\varepsilon)}}}~.
\end{equation}
Repeat the following steps for $g=0,1,2,\ldots$:
\begin{enumerate}
\item prepare the equal superposition state $\ket{\psi}=\frac{1}{\sqrt{N}}\sum_{x=1}^N\ket{x}$ on a quantum register $\mathbb{C}^N$,
\item perform $k_g:=\left\lfloor\alpha_g(\varepsilon)\frac{\pi}{4}\sqrt{N}\right\rfloor$ Grover iterations,
\item measure in the standard basis, verify the result using one oracle invocation; if correct then store in classical output register.
\end{enumerate}
\end{algorithm}

The following theorem proves that Algorithm~\ref{algounknownpnotexcluding} is fault-ignorant, i.e.~finds the marked element independently of the actual noise level (in particular, the algorithm does not need $p$ as an input), and gives a bound on its runtime which, however, depends on the actual noise level.
\begin{theorem}[Fault-ignorance of Algorithm \ref{algounknownpnotexcluding}]\label{thmunknownpwithoutexcluding}
Let $p\in[0,1]$ be the actual noise level (i.e.\ noise $D_p(\varrho):=(1-p)\varrho+pT(\varrho)$ with any quantum channel $T$) acting on the quantum register when executing Algorithm~\ref{algounknownpnotexcluding}, and let $N\ge100$ and $\varepsilon\in(0,1/2]$. Then Algorithm~\ref{algounknownpnotexcluding}, with $c=10$, finds the marked element after at most
\begin{equation}\label{stepsinfirstfaulttolerantalgo}
100\left(Np+\sqrt{N}\right)\log\frac{1}{\varepsilon}
\end{equation}
oracle queries except with probability at most $\varepsilon$.
\end{theorem}If one wants Algorithm \ref{algounknownpnotexcluding} to be fault-ignorant only for noise levels $p\in[0,0.1]$, and if one presupposes reasonable values $N\geq1000$ and $\varepsilon\in(0,0.1]$, then one can guarantee the constant prefactor to be $20$ instead of $100$ as in (\ref{stepsinfirstfaulttolerantalgo}), when using $c=4.5$; cf.\ Appendix \ref{mainthm}.

As the proof of Theorem \ref{thmunknownpwithoutexcluding} is rather technical, we present here only a sketch; for details, see Appendix \ref{mainthm}.
\begin{proof}[Sketch of the proof]As the noise acts symmetrically with respect to the different oracles and since the Grover steps of Algorithm~\ref{algounknownpnotexcluding} are symmetric as well, the success events for different rounds $g$ are independent, so that we will be able to upper bound the failure probability after round $g^*$:
\begin{equation}\label{failprobabinfirstlowerbound}
\fail~=~\prod_{g=0}^{g^*}\left(1-p_s(N,k_g,p)\right)~\le~\exp\left\{-\sum_{g=0}^{g^*}(1-p)^{k_g}\sin^2\left((2k_g+1)\arcsin\frac{1}{\sqrt{N}}\right)\right\}~.
\end{equation}
When the sum $\sum_{g=0}^{g^*}$ in the last expression is greater than $\log(1/\varepsilon)$ then we can guarantee the failure probability to be at most $\varepsilon$, as desired. To get the statement about the number of oracle calls, we upper bound it by
\begin{equation}\label{numoforaclecallsfirstlowerbound}
\sum_{g=0}^{g^*}(k_g+1)~
\le~g^*+1+\frac{\pi}{4}\sqrt{N}+\frac{\pi}{2}\sqrt{N}\left(c\log\frac{1}{\varepsilon}\right)\sqrt{1+\frac{g^*}{c\log(1/\varepsilon)}}~.
\end{equation}

The proof of Theorem \ref{thmunknownpwithoutexcluding} now consists in showing that there exists a number $g^*$ (of rounds) such that the failure probability (\ref{failprobabinfirstlowerbound}) is at most $\varepsilon$, while the number of oracle calls (\ref{numoforaclecallsfirstlowerbound}) is at most (\ref{stepsinfirstfaulttolerantalgo}). Similar to our analysis leading up to (\ref{eq:firstbound}), this argument will be split into three different cases: for $p\le\pi/(4\sqrt{N})$ the first few rounds ($g=0,1,2,\ldots$) are the important ones; for $p\ge 0.3$ we take into account only the rounds with $k_g=0$ (i.e.~large $g$); and for $\pi/(4\sqrt{N})\le p\le 0.3$ our proof relies on an intermediate regime of $g$. Details are given in Appendix \ref{mainthm}.
\end{proof}

Theorem \ref{thmsymmalgomemoryless} actually shows that the runtime of Algorithm \ref{algounknownpnotexcluding} (which we just proved to be at most Eq.\ (\ref{stepsinfirstfaulttolerantalgo})) is \emph{optimal up to a constant}: To see this, note that Algorithm \ref{algounknownpnotexcluding} is contained in the class of algorithms to which the bound from Theorem \ref{thmsymmalgomemoryless} on the number of oracle queries $q_N\gsim\frac{1}{8}Np\log\frac{1}{\varepsilon}$ applies. For any fixed noise level $p>0$ and up to a constant factor, this equals the upper bound (\ref{stepsinfirstfaulttolerantalgo}) on the number of queries needed by Algorithm \ref{algounknownpnotexcluding}. In particular, even if one does \emph{not} know the actual noise level in advance, one only loses a constant factor in the number of queries, compared to the runtime in case of known $p$ (given in Eq.\ (\ref{eq:firstbound})).

We re-emphasize here the last point from Subsection \ref{subsectrepeatbuildingblockknownp}, that for small enough but constant noise levels $p$ and in the limit $N\to\infty$, the quantum Algorithm \ref{algounknownpnotexcluding} needs \emph{fewer} oracle calls than even the best classical algorithm in a noiseless environment.

As the runtime depends on an unknown parameter, it is necessary to have the ability to stop the algorithm as soon as the result is found. Theorem \ref{thmunknownpwithoutexcluding} then states a ``probabilistic bound'' on the number of oracle uses up to the point when the element is found; this bound is probabilistic in the sense that in a fraction of at most $\varepsilon$ of all runs of Algorithm \ref{algounknownpnotexcluding}, the actual runtime may exceed this bound.

When considering more general algorithms, namely for which the events of failure in different rounds are not independent, the derivation of the lower bound on the necessary number of queries from Theorem \ref{thmsymmalgomemoryless} is no longer valid. This dependency can arise either from asymmetric noise or from an asymmetry in the algorithm itself. Indeed, it is useful to consider such ``asymmetric algorithms'':~already classically one can find the marked element using $\lceil N(1-\varepsilon)\rceil$ queries, except with error probability $\varepsilon$, by simply testing a subset of $\lceil N(1-\varepsilon)\rceil$ elements using one oracle call each. This feature of not considering previously falsified items again is absent from Algorithm \ref{algounknownpnotexcluding} whose runtime may therefore exceed that of classical search, through the factor $\log(1/\varepsilon)$ in Theorem \ref{thmsymmalgomemoryless}.

The algorithms considered in Section \ref{sec:withmemory} will make use of this asymmetry, which can also be conceived of as conditioning the actions in future rounds on previous measurement outcomes that are being stored in a classical memory. This will be done by incorporating a noiseless classical memory which we will allow the algorithms to use in a limited way, namely by excluding oracle indices that have been falsified in previous rounds.

\section{Search algorithms employing noiseless classical memory}\label{sec:withmemory}

\subsection{Search with exclusion}\label{firstsearchwithexclusionsubsection}
Classical search algorithms can find the marked element with maximal failure probability $\varepsilon$ using at most $\lceil N(1-\varepsilon)\rceil$ steps, by excluding falsified oracle indices. Here we aim to achieve an upper bound of $N$ on the runtime --- independently of $\varepsilon$ and of $p$ --- for our quantum algorithms as well, whereas in Section \ref{sec:memoryless} we have only presented algorithms whose runtime may exceed $N$ parametrically due to the factor $\log(1/\varepsilon)$, e.g.~in (\ref{stepsinfirstfaulttolerantalgo}).

On a quantum computer a random choice may be implemented by preparing the equal superposition state $\ket{\psi}$ over a subset of basis states followed by a measurement in that basis. This in turn can be viewed as a Grover search with zero iterations (cf.~Subsection \ref{setupnoiselessandmodels}). This leads to the idea of replacing the uniform random choices by proper Grover searches (each including several Grover steps plus a concluding measurement) over decreasing subsets, i.e.~$\{1,\ldots,N\}\setminus\{i_1,\ldots,i_{m'}\}$ after round $m'$. For this, the classical noiseless memory that previously stored only the correct search outcome, will be exanded by a register of $2^N$ states ($N$ bits) to mark the previously excluded items.

We shall now develop and sketch a search algorithm based on this idea of excluding previously tested elements; the following procedure is applicable if the noise level $p$ is known beforehand. If one fixes the number $N$ of database entries, the noise parameter $p$ and the target error bound $\varepsilon$, then the question is how to choose the number of iterations in each round in order to consume the least number of queries. Suppose that in the $i$th round we perform $k_i$ queries and we do $m$ rounds in total. Then the number of queries is $\sum_{i=1}^{m}(k_i+1)$, while the probability of not finding the marked element is at most $\prod_{i=1}^m\left(1-p_s(N-i+1,k_i,p)\right)$; thus, the minimal number of queries for which one can guarantee success (up to failure probability $\varepsilon$) in the general noise model is
\begin{equation}
\min\left\{\sum_{i=1}^{m}(k_i+1)\,\Bigg|\,m\in\mathbb{N},\,i_1,\ldots,i_m\in\mathbb{N},\,\prod_{i=1}^m\left(1-p_s(N-i+1,k_i,p\right))\le\varepsilon\right\}\label{mindynprogramming}~,
\end{equation}
e.g.~letting $p_s(N,k,p)\equiv(1-p)^k\sin^2((2k+1)\arcsin(1/\sqrt{N}))$ equal the lower bound in (\ref{lowerboundprobs}) (alternatively, (\ref{exactsuccessprobabfordepol})). For given $(N,p,\varepsilon)$, the minimum (\ref{mindynprogramming}) and the corresponding sequence $\{k_i\}$ of Grover steps can be found by dynamic programming. Clearly, $\lceil N(1-\varepsilon)\rceil$ is an upper bound on the number of oracle calls and for any fixed $p_0>0$ we can bound this as $\lceil N(1-\varepsilon)\rceil\leq Np(1-\varepsilon)/p_0$ for all $p\geq p_0$. Similarly, by (\ref{eq:firstbound}) or Theorem \ref{thmunknownpwithoutexcluding}, for $\varepsilon\geq\varepsilon_0>0$ there is also an upper bound of the form $cNp(1-\varepsilon)$, since $\log(1/\varepsilon)\leq c(1-\varepsilon)$ where $c$ is determined by $\varepsilon_0$. Hence, an upper bound on the runtime of the form $c'Np(1-\varepsilon)$ holds for the complement of any neighbourhood of $(p,\varepsilon)=(0,0)\in[0,1]^2$, at least asymptotically for $N\to\infty$.

In following subsection we simplify the above algorithm, based on typical behaviour of the sequences $\{k_i\}_i$ found in numerical experiments.

\subsection{Fault-ignorant quantum search with exclusion}\label{subsectionfaultignorantwithexclusion}
In this subsection we present a more explicit algorithm to solve the search problem in the \emph{fault-ignorant setting}, i.e.~an algorithm which can be specified and works even for \emph{unknown} noise level $p$, using the exclusion described above to obtain faster runtime (cf.~also Fig.~\ref{fig:alg2}):
\begin{algorithm}[Quantum search with exclusion]\label{algounknownpexcluding}
For suitably chosen $c>0$, define $S_0:=\{1,\ldots,N\}$ and
\begin{equation}\label{alphaprescriptalg2}
\alpha_g(\varepsilon)~:=~\frac{1}{\sqrt{1+\frac{g}{c\log(1/\varepsilon)}}}~.
\end{equation}
Repeat the following steps for $g=0,1,2,\ldots$:
\begin{enumerate}
\item prepare the equal superposition state $\psi_g$ over the set $S_g$,
\item perform $k_g$ Grover iterations (with $I-2\ketbra{\psi_g}{\psi_g}$ as reflection), where
\begin{equation}\label{selectstep2ofalgo2}
k_g~:=~\left\{\begin{array}{ll}
\displaystyle\left\lfloor\alpha_g(\varepsilon)\frac{\pi}{4}\sqrt{N-g}\right\rfloor~, &~~~\text{if }k_0+k_1+\cdots+k_{g-1}+g\le (1-\varepsilon)N~, \\
0~, &~~~\text{otherwise}~,
\end{array}\right.
\end{equation}
\item measure in the standard basis, verify the result $r_g$ using one oracle invocation, store if correct,
\item let $S_{g+1}:=S_g\setminus\{r_g\}$.
\end{enumerate}
\end{algorithm}

\begin{figure}
\begin{center}
\includegraphics{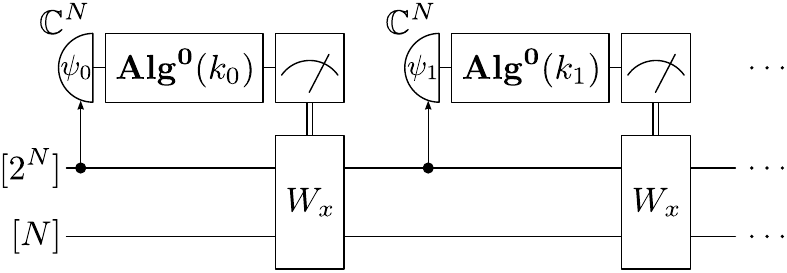}
\end{center}
\caption{\label{fig:alg2}Algorithm \ref{algounknownpexcluding} uses exclusion in searching for an element in consecutive rounds of $k_0,k_1,k_2,\ldots$ Grover steps each, supplemented by one verification query. Each round starts by preparing the equal superposition state $\psi_g$ of the previously not excluded elements, noted in the classical memory $[2^N]$, and is concluded by a measurement in the computational basis. The output is then verified against the oracle (list look-up) and stored in the classical noiseless memory $[N]$ if the element is found and marked in the memory $[2^N]$ if the round was unsuccessful ($W_x$).}
\end{figure}

Similarly as Theorem \ref{thmunknownpwithoutexcluding} for Algorithm \ref{algounknownpnotexcluding}, the following theorem proves fault-ignorance of Algorithm \ref{algounknownpexcluding} and provides a bound on its runtime:
\begin{theorem}[Fault-ignorance of Algorithm \ref{algounknownpexcluding}]\label{thmknownpwithoutexcluding}
Let $p\in[0,1]$ be the actual noise level (i.e.\ noise $D_p(\varrho):=(1-p)\varrho+pT(\varrho)$ with any quantum channel $T$) acting on the quantum register when executing Algorithm \ref{algounknownpexcluding}, and let $N\geq100$ and $\varepsilon\in(0,1/2]$. Then Algorithm \ref{algounknownpexcluding}, with $c=10$, finds the marked element after at most
\begin{equation}
\label{faultignorantbound}
\min\left\{100\left(Np+\sqrt{N}\right)\log\frac{1}{\varepsilon},\,\,2(1-\varepsilon)N+\sqrt{N}\right\}
\end{equation}
oracle queries except with probability at most $\varepsilon$.
\end{theorem}
\begin{proof}The success probability in each round of Algorithm \ref{algounknownpexcluding} is at least as large as in Algorithm \ref{algounknownpnotexcluding} because we are excluding elements; thus, the runtime of Algorithm \ref{algounknownpnotexcluding} puts an upper bound on the runtime by (\ref{stepsinfirstfaulttolerantalgo}). Before that, however, Algorithm \ref{thmunknownpwithoutexcluding} may switch to testing (and excluding) elements in random order (the second selector in (\ref{selectstep2ofalgo2})); this switch happens after at most $(1-\varepsilon)N+(\pi/4)\sqrt{N}+1$ oracle calls, and after the switch Algorithm \ref{algounknownpexcluding} needs at most $(1-\varepsilon)N+1$ additional calls to find the marked element except with probability $\varepsilon$.\end{proof}

The constant $100$ in (\ref{faultignorantbound}) can be improved to $20$ for a restricted range of parameters, following the remark below Theorem \ref{thmunknownpwithoutexcluding}.

Again, similar to Theorem \ref{thmsymmalgomemoryless}, we can show that, despite restricting to Grover's specific steps, Algorithm \ref{algounknownpexcluding} is essentially optimal within a wider class of algorithms. Namely, we extend the class of algorithms considered in and before Theorem \ref{thmsymmalgomemoryless} in such a way that, instead of requiring independence of the failure probabilites in different rounds, we assume that after each unsuccesful round we exclude the tested element and thereby reduce the search space as well as the state space of the computation. We also need to ensure that failure probabilities are multiplicative, which is the case e.g.~if both the noise and the algorithm treat the search elements uniformly (this in particular applies to Grover iterations and partial depolarizing noise, see Eq.\ \ref{partiallydepolnoise}).

This wider class of algorithms is qualitatively different from the algorithms considered in Theorem \ref{thmsymmalgomemoryless}, as it now contains algorithms succeeding with $(1-\varepsilon)N\le N$ oracle calls, \emph{independently} of $p$, for example the classical verification-and-exclusion algorithm described at the beginning of Sec.\ \ref{firstsearchwithexclusionsubsection}. This is reflected in the fact that the lower bound in the following theorem never exceeds $N$, unlike the bounds on $q_N$ in the memoryless setting from Theorem \ref{thmsymmalgomemoryless}.
\begin{theorem}[Lower runtime bound on exclusion algorithms]\label{thmlowerboundwithmemory}For any quantum search algorithm (that may or may not have the noise level $p$ as an input) satisfying the above constraints and whose quantum register is subject to depolarizing noise (see Eq.\ (\ref{partiallydepolnoise})) with fixed strength $p\in(0,1)$, the number $q_N$ of queries to find the marked element up to fixed failure probability $\varepsilon\in(0,1)$ is lower bounded as
\begin{equation}\label{eqninmemorythmlowerbound}
q_N~\ge~\frac{N}{1+\frac{8C_N}{p\log(1/\varepsilon)}}~,
\end{equation}
for some sequence $C_N=C_N(\varepsilon,p)$ with $\lim_{N\to\infty}C_N=1$.
\end{theorem}
\begin{proof}
We can assume that $\limsup_{N\to\infty}q_N/N\le1-\varepsilon$, because there does exist an algorithm with this limit being $1-\varepsilon$ (see above). For now, fix $N$, and by letting $N\to\infty$ later we will make sure that all following expressions are well-defined (e.g.~no logarithms of negative arguments occur, etc.).

Let the number of queries in round $g$ ($0\le g\le r$) be $k_g+1$ (i.e.~including the verification-exclusion step), so $r<q_N=\sum_{g=0}^r(k_g+1)$. For the success probabilty in round $g$ we have by (\ref{moretrivialimplication}) and (\ref{8kplus1overp})
\begin{equation}
p_s(N-g,k_g,p)~\le~\frac{h(p,k_g)}{N-g}~\quad~\text{with}~\quad~h(p,k)~\le~1+\frac{8}{p^2}~\quad\text{and}~\quad h(p,k_g)~\le~\frac{8(k_g+1)}{p},
\end{equation}
where the latter inequality implies that $h(p,k)$ is bounded independently of $k$ (and of $N$), since $p$ is given.
Thus we can lower bound the failure probability (using $\log(1-x)\ge-x/(1-x)$):
\begin{align}
\varepsilon~
&\ge~\prod_{g=0}^r\left(1-p_s(N-g,k_g,p)\right)~\ge~\exp\left\{\sum_{g=0}^r\log\left(1-\frac{h(p,k_g)}{N-r}\right)\right\}~\ge~\exp\left\{ \sum_{g=0}^r-\frac{h(p,k_g)}{N-r}\frac{1}{1-\frac{h(p,k_g)}{N-r}}\right\}\nonumber\\
&  \ge~\exp\left\{-\sum_{g=0}^r\frac{8(k_g+1)}{(N-r)p}\,\,\left[1-\frac{1}{N-r}\left(1+\frac{8}{p^2}\right)\right]^{-1}\right\}~\ge~\exp\left\{-\frac{8}{p}\frac{q_N}{N-q_N}C_N\right\}~,\label{longchaininmemoryopimalthem}
\end{align}
where we used $g\le r\le q_N$ from above and defined $C_N:=\left[1-\frac{1}{N-q_N}\left(1+\frac{8}{p^2}\right)\right]^{-1}$, which converges to $1$ as $N\to\infty$.  Inverting (\ref{longchaininmemoryopimalthem}) to get an explicit bound on $q_N$ finally gives (\ref{eqninmemorythmlowerbound}).
\end{proof}With the usual conventions in treating $1/0$ and $1/\infty$, Theorem \ref{thmlowerboundwithmemory} is valid for all $p,\varepsilon\in[0,1]$. Similar to Eq.\ (\ref{finiteNbound}) in Theorem \ref{thmsymmalgomemoryless}, one could explicitly specify a sequence $C_N$ in the bound (\ref{eqninmemorythmlowerbound}) which would however complicate the expression.

\bigskip

In summary, Algorithm \ref{algounknownpexcluding}, which uses the exclusion strategy and Grover iterations, is \emph{fault-ignorant} and Theorem \ref{thmknownpwithoutexcluding} provides an upper bound on its runtime. Conversely, Theorem \ref{thmlowerboundwithmemory} provides a lower bound on the number of oracle calls for any symmetric fault-ignorant algorithm using the exclusion strategy. And the inequalities from Lemma \ref{memoryoptimaliylemma} (see Appendix \ref{app:technical}) show that Algorithm \ref{algounknownpexcluding} is basically optimal within this class of algorithms, in the sense that for any $p,\varepsilon\in(0,1)$ its runtime is at most a constant factor (independent of $p$ and $\varepsilon$) above the lower bound from Theorem \ref{thmlowerboundwithmemory}.

And even stronger, the lower bound on the runtime in Theorem \ref{thmlowerboundwithmemory} is proven for algorithms that may ``know'' the noise level $p$ as one of their inputs (such as the algorithm resulting from (\ref{mindynprogramming})), whereas our Algorithm \ref{algounknownpexcluding} basically saturates this lower bound \emph{without} actually being dependent on the actual noise level $p$; the latter feature is the characteristic of \emph{fault-ignorant algorithms}. Thus, not knowing the noise level inflicting upon the quantum computation extends the runtime at most by a constant factor, which was observed in the memoryless setting following Theorem \ref{thmunknownpwithoutexcluding} as well.

\subsection{Our search algorithms, and comparision to other work \cite{RegevSchiff,SBW,temme}}\label{moregeneralmemoryalgos}
In Section \ref{sec:memoryless} we did not allow for a classical memory (except to store the correct output), whereas in Section \ref{sec:withmemory} we allowed a noiseless classical register in order to exclude falsified items from future search rounds. This is obviously not the most general class of algorithms. One may for example perform non-projective measurements which could result in a non-uniform distribution over oracles (cf.~\cite{ashley}) after the measurement. Or one may abandon the division into ``rounds'' altogether, and rather use the noisy quantum register and noiseless classical memory in a more general way (cf.~Appendix \ref{sec:faultignorance}). While these possibilities are rather vague, at least in the noiseless case ($p=0$) Grover's algorithm is exactly optimal \cite{Zalka}. A similar general proof eludes us in the noisy case considered in this paper; Theorems \ref{thmsymmalgomemoryless} and \ref{thmlowerboundwithmemory} give such a bound under more restrictive qualifications.

Nevertheless, the results obtained here may suggest that any nonzero noise level $p>0$ (in our noise models, cf.~Appendix \ref{sec:noisemodels}) prolongs the runtime beyond the noiseless lower bound (in \cite{Zalka}), necessitating it to be proportional to the number of search items as $N\to\infty$. However, for small but constant noise level $p>0$, the runtime bound $\lsim{Np\log(1/\varepsilon)}$ on our algorithms (cf.~Theorems \ref{thmunknownpwithoutexcluding} and \ref{thmknownpwithoutexcluding}) can be far below the $N(1-\varepsilon)$ oracle calls required by the best noiseless classical algorithm. In this regard, see \cite{qrefrigpaper} for a treatment of \emph{locally} acting noise and questions about optimality in this case.

\bigskip

Similar in spirit to our Theorems \ref{thm8kplus1overp}, \ref{thmsymmalgomemoryless} and \ref{thmlowerboundwithmemory}, a lower bound of $\gsim Np/(1-p)$ on the runtime of general noisy quantum search algorithms was obtained in \cite{RegevSchiff}, whose faulty oracle model is somewhat similar to our partial depolarizing noise (\ref{partiallydepolnoise}) (with roughly the same noise parameter $p$; they fixed $\varepsilon\simeq1/10$); see also \cite{temme} for a continuous-time analogue of this result. One difference is that, in Theorems \ref{thmsymmalgomemoryless} and \ref{thmlowerboundwithmemory}, we allow error-free (e.g.~classical) verification steps, whereas every oracle use in \cite{RegevSchiff} is potentially faulty, leading to a diverging bound as $p\to1$. Also, \cite{RegevSchiff,temme} does not include a noiseless classical memory. Due to these extensions, our lower bounds are restricted to ``symmetric'' algorithms consisting of ``rounds'', whereas \cite{RegevSchiff,temme} applies to \emph{all} algorithms within their memoryless setting. These works do not consider achievability of the bound.

The work \cite{SBW} specifically investigated Grover's algorithm under phase noise (see also \cite{temme}), again somewhat analogous to our noise model (\ref{partiallydepolnoise}). It was observed that Grover's algorithm gives an advantage only if it runs for $k\lsim1/p$ steps, and it was hinted that at this time one may perform a measurement and start a new Grover round. In Sections \ref{basicbuildingblocksubsection} and \ref{subsectrepeatbuildingblockknownp}, we give more rigorous arguments (and prefactors) for the scaling $k\sim1/p$, also for algorithms not necessarily consisting of Grover steps. Our Algorithms \ref{algounknownpnotexcluding} and \ref{algounknownpexcluding} do indeed use the division into Grover rounds, but they even function \emph{fault-ignorantly}. The avoidance of active error correction \cite{preskill} is advocated by \cite{SBW} as well.

A more technical difference of our work to most of the literature is that we consider the \emph{average} success probability of an algorithm, i.e.\ averaged over all $N$ oracles with equal weight, whereas the literature most often only investigates the \emph{minimum} success probability of any of the $N$ oracles. This makes our lower bounds stronger than the ones obtained in the literature. (As our constructive algorithms are symmetric, the minimum and average success probabilities coincide for those.)

\section{Conclusion}\label{conclusionsection}
In this paper we have investigated the idea of \emph{fault-ignorant quantum algorithms}. Such algorithms should output the correct result even in the presence of noise of potentially unknown strength, in such a way that the actual noise level $p$ may affect the runtime it takes to arrive at the correct answer (up to some specified failure probability $\varepsilon$), but should not affect that fact that the correct answer is found eventually. This approach allows to reduce the required spatial circuit sizes, for example in comparison to using full-scale quantum error correction, however at the expense of increased runtime.

Following this general idea, we have provided fault-ignorant algorithms for quantum searching that function under depolarizing or dephasing noise of unknown strength $p$. We find the ``quadratic speedup'' to be achievable only for low decoherence rates $p\lsim 1/\sqrt{N}$. Otherwise, our best algorithm's runtime scales asymptotically like $\min\{Np\log 1/\varepsilon,\,N(1-\varepsilon)\}$ as $N\to\infty$. This is linear in $N$, but for low enough noise levels $p$ it nevertheless outperforms the optimal classical search algorithm. Our algorithms may thus be useful for initial uses of quantum computing, when unlimited scalability of the size of quantum computers is not yet be achievable due to technological limitations.

We moreover proved that, up to a constant factor, our algorithms runtimes are optimal within wide classes of noisy quantum search algorithms. Remarkably, for the searching task, it turned out that ignorance of the actual noise level will extend the runtime by only a constant factor compared to the case of known noise level $p$.

Due to the novelty of the approach, our algorithms and lower bounds leave questions for further research. On the side of concrete algorithms, one may ask for them to be independent not only of the noise level $p$ but also of the desired accuracy $\varepsilon$; then one could continue running the algorithm for longer time to increase the success probability or accuracy.

Concerning lower bounds on the complexity of noisy quantum search, it would be worthwhile to establish an analogue of Theorem \ref{thm8kplus1overp} for the case of local noise models or even partial dephasing or general partial entanglement-breaking noise. The latter would immediately extend the validity of our lower bound in Theorem \ref{thmlowerboundwithmemory} to the class of quantum algorithms that use a noiseless classical register in an arbitrary way and need not be divided into ``quantum rounds''. In a similar vein, it may also be possible to prove that the essentially optimal runtime for quantum searching under partial depolarizing noise, which we mainly investigated, can always be achieved by an algorithm divided into such rounds (see Sec.\ \ref{moregeneralmemoryalgos}). If not, it would be very interesting to find fault-ignorant algorithms that are not of this simple check-and-repeat form.

Finally, it would be desirable to investigate whether and how the fault-ignorant idea could possibly be applied to computational models other than the quantum circuit model. This would in particular be desirable in computational models for which quantum error correction techniques are less developed, such as adiabatic quantum computing, and where other methods to achieve fault tolerance are needed. Generally, we hope that, beyond unstructured search, the fault-ignorant idea will be fruitfully applied to algorithmic tasks, such as sampling algorithms.

\bigskip

{\bf{Acknowledgments:}} We thank the two anonymous referees for their insightful comments which helped to improve the paper. This research was initiated at a workshop of the FP7 project COQUIT, which also supported P.\ Vrana and D.\ Reitzner. D.\ Reeb was supported by the Marie Curie Intra-European Fellowship QUINTYL. M.\ Wolf acknowledges support from the CHIST-ERA/BMBF project CQC and the Alfried Krupp von Bohlen und Halbach-Stiftung.

\appendix

\section{Fault-ignorance --- a mathematical framework}\label{sec:faultignorance}

The aim of this appendix is to provide a rigorous mathematical definition of fault-ignorant computing (see Section \ref{moreformaldescription} for a less formal discussion). We do this in a way which enables to include a fairly broad class of algorithmic problems into this framework in a unified manner, while keeping the definition reasonably simple. The definitions are supposed formalize algorithms that do not need to know the actual noise level in order to accomplish their task --- they should be \emph{ignorant} of the noise. A fault-ignorant algorithm should be robust enough to provide the answer (up to some specified failure probability) under any level of noise, the latter affecting only its runtime.

In our formalization, we want to allow the desired and the actual output of the algorithm to be probabilistic (as is usual in sampling and quantum simulation problems), and to depend on an input (as for example in computational problems) and on an oracle (as in search problems). Given the discrete-time nature of the computation as well as of our noise models (cf.~Appendix \ref{sec:noisemodels}), it is necessary to explicitly refer to an allowed class of quantum operations (the set $S$ in the following definition). This can be done most conveniently if one also specifies the (spatial) resources available for performing the computation, i.e.~the size of the quantum computer available or of any accompanying classical register. We thus view the specification of the size of the available quantum register as part of the task to be solved; and indeed, since early realizations of quantum computers will be limited in the number of qubits, it will be a part of the challenge to solve a desired task on the \emph{available} hardware, esp.\ under noise influence because full-scale quantum error correction may be prohibited.  Further, we consider only the quantum register to be noisy, whereas noiseless classical memory is today a reasonable technological assumption.

In light of this, we propose the following definitions, which we explain and supplement by examples afterwards.
\begin{definition}[Noisy quantum computational task]\label{definitionnoisytask}
A \emph{noisy quantum computational task} is a tuple $(X,I,O,f,\mathcal{H},\varrho,D,s,S)$ where
\begin{itemize}
\item $X$, $I$, $O$ are sets,
\item $f\in\mathbb{R}^{X\times I\times O}$ is a stochastic matrix, i.e.~has nonnegative entries and for any $x\in X$ and $i\in I$ we have $\sum_{o\in O}f_{xio}=1$,
\item $\mathcal{H}$ is a Hilbert space,
\item $\varrho_\cdot:I\to\B(\mathcal{H})$ is a function with density operators as values,
\item $D_\cdot:[0,1]\to\CPT(\B(\mathcal{H}))$ is a function with quantum channels on $\B(\mathcal{H})$ as values,
\item $s\in\mathbb{N}$,
\item $S~\subseteq~\CPT(\B(\mathcal{H})\otimes\mathbb{C}^s\otimes\mathbb{C}^{O}\otimes\mathbb{C}^2)^{{X}}~=~\left\{\,T:{X}\to \CPT(\B(\mathcal{H})\otimes\mathbb{C}^s\otimes\mathbb{C}^{O}\otimes\mathbb{C}^2)\,\right\}$.
\end{itemize}
\end{definition}

\begin{figure}
\begin{center}
\includegraphics{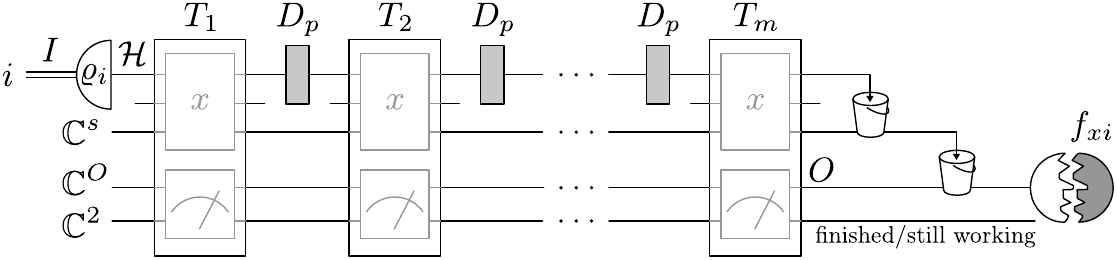}
\end{center}
\caption{\label{fig:ignorance}A fault-ignorant algorithm (Definition \ref{definefaultignorantalgo}), specified to solve a noisy quantum computational task (Definition \ref{definitionnoisytask}): An input state $\varrho_i$ is affected in turn by devised operations $T_j$ (which may include an oracle indexed by $x$, or other coherent operations, or measurement/verification procedures that store information in noiseless classical registers $\mathbb{C}^s$, $\mathbb{C}^O$, $\mathbb{C}^2$) and noise $D_p$ acting on the quantum register. After some number of steps, the probability distribution of the output $o\in O$ should approximate the desired distribution $f_{xi}$, up to error $\varepsilon$.}
\end{figure}

We interpret $X$ as the set labelling the different oracles, $I$ and $O$ as the sets of possible inputs and outputs, respectively (see also Fig.~\ref{fig:ignorance}). For a task which does not make use of an oracle, we let $X$ be any singleton set, and similarly, if the computation does not need an input, we let $|I|=1$. The stochastic matrix $f$ describes the desired distribution on the output set depending on the input and the oracle. The computation is performed using the Hilbert space $\mathcal{H}$ and a classical memory of $s$ states, with the output being written into an additional classical register with $|O|$ states, corresponding to the possible output states in $O$. The additional classical bit ${\mathbb C}^2$ is to have value $1$ iff the algorithm wants to signal that the result is available in the register ${\mathbb C}^O$. The reason for this is that in the fault-ignorant setting the runtime depends on an unknown parameter (namely $p$, see below), and therefore the algorithm needs a way to tell whether the computation is already done, without destroying the quantum state.

The map $\varrho_\cdot:I\to\B(\mathcal{H})$ plays the role of input encoding in the sense that the physical initial state $\varrho_i$ on the register $\B(\mathcal{H})$ represents the abstract input value $i\in I$. The quantum register $\mathcal{H}$ is subject to noise modeled by the quantum channels $D_p$ (as specified in Eq.\ (\ref{defineoutputstate}) below) depending on a parameter $p\in[0,1]$, which we think of as a strength parameter.

Finally, the set $S\subseteq\CPT(\B(\mathcal{H})\otimes\mathbb{C}^s\otimes\mathbb{C}^{O}\otimes\mathbb{C}^2)^{{X}}$ represents the set of allowed elementary steps. An element in this subset is understood as a quantum channel depending on the oracle $x\in X$, whereby the quantum channel acts on the quantum register $\B(\mathcal{H})$ as well as on the classical (diagonal) registers $\mathbb{C}^s$, $\mathbb{C}^{O}$ and $\mathbb{C}^2$ described above (our specification below will be such that all these classical registers are initialized in the state $|0\rangle\langle0|$ at the start of an algorithm). This gives a way to impose conditions on how ``complicated'' the elementary operations are, e.g.~in terms of oracle use or locality requirements (see examples below), and at the same time it maps the abstract oracle $x\in X$ to actual physical transformations $T(x)$ it may perform.

\begin{definition}[Fault-ignorant algorithm]\label{definefaultignorantalgo}
A \emph{fault-ignorant algorithm} solving the noisy quantum computational task $(X,I,O,f,\mathcal{H},\varrho,D,s,S)$ is a family $((T^\varepsilon_n)_n)_{\varepsilon\in(0,1)}$ of finite or infinite sequences with $T^\varepsilon_n\in S$ such that for all $\varepsilon\in(0,1)$ and for all $p\in[0,1]$ the value
\begin{equation}
t^\varepsilon(p)~:=~\min\left\{\,t_0\in\mathbb{N}\,\,\Bigg|\,\,\forall t\geq t_0:\forall x,i:\,\,\frac{1}{2}\left\|\hat{f}_{xi}-s^{\varepsilon,t}_{xi}(p)\right\|_1\le\varepsilon\,\right\}\label{defineminimalsteps}
\end{equation}
is finite, where $\hat{f}\in\mathbb{R}^{X\times I\times O\times\{0,1\}}$ is defined by $\hat{f}_{xio0}:=0$ and $\hat{f}_{xio1}:=f_{xio}$ for $x\in X$, $i\in I$, $o\in O$, and
\begin{equation}
s^{\varepsilon,t}_{xi}(p)~:=~\Tr_{\B(\mathcal{H})\otimes\mathbb{C}^s}\left[T^\varepsilon_t(x)D_pT^\varepsilon_{t-1}(x)D_p\ldots D_pT^\varepsilon_2(x)D_pT^\varepsilon_1(x)\left(\varrho_i\otimes\ketbra{0}{0}\otimes\ketbra{0}{0}\otimes\ketbra{0}{0}\right)\right]\label{defineoutputstate}
\end{equation}
is a probability distribution on $O\times\{0,1\}$.
\end{definition}

Thus, the sequence of operations in (\ref{defineoutputstate}) models a $t$-step noisy quantum computation, in the sense that between any two elementary operations from $S$ a noise channel $D_p$ is to be applied on the quantum register $\B(\mathcal{H})$ (Fig.~\ref{fig:ignorance}). The sequence $(T^\varepsilon_n)_n$ itself describes the algorithmic operations, which may depend on the required accuracy $\varepsilon$, i.e.~on the maximally tolerable distance from the desired output distribution $\hat{f}_{xi}$, cf.~(\ref{defineminimalsteps}).

The requirement for $t^\varepsilon(p)$ to be finite for \emph{any} $p$, even though the algorithm $((T^\varepsilon_n)_n)_{\varepsilon\in(0,1)}$ does \emph{not} depend on $p$, justifies the term \emph{fault-ignorant algorithm}. The condition ``$\forall t\geq t_0$'' in (\ref{defineminimalsteps}) requires the result to be available in the classical memory at any later time when the outside agent, ignorant of the noise level $p$ and thus of the necessary computation time $t^\varepsilon(p)$, may check the ${\mathbb C}^2$ flag to see whether the computation has already finished and want to read out the result. Note that Definition \ref{definefaultignorantalgo} does \emph{not} put any requirements on the efficiency of the algorithm, which however in some circumstances may be quantified by $t^\varepsilon(p)$, i.e.~the minimal number of invocations of $T_k^\varepsilon(x)$ (e.g.~oracle calls); see Section \ref{moreformaldescription}.

\bigskip

We now illustrate the definitions above by two examples.
\begin{example}[Quantum search]As an example we now show how the noisy quantum search problem considered in Sections \ref{sec:memoryless} and \ref{sec:withmemory} fits into this framework. In this case we have a set of $N$ oracles $X=\{1,\ldots,N\}$, and the algorithm is required to identify the oracle, so $O=X$. Since no input is needed, we set $I=\{0\}$. Now the function to be computed is deterministic, so $f$ will be a $0$-$1$ matrix, more specifically $f_{xio}=\delta_{xo}$. The Hilbert space we use is $\mathcal{H}=\mathbb{C}^N\otimes\mathbb{C}^M\otimes\mathbb{C}^2$ for some $M$ setting the size of the ancillary quantum system and $\mathbb{C}^2$ standing for the ancillary system used by the oracle, cf.\ Eq.\ (\ref{oraclebasic}). The noise acting on it is for example partial depolarizing, $D_p(\varrho)=p(\Tr\varrho)\frac{I_{2NM}}{2NM}+(1-p)\varrho$. Since there is no input, $\varrho_0$ is just any fixed initial state, e.g.~we may take $\varrho_0=\frac{I_{2NM}}{2NM}$. In the version without classical memory we set $s=1$ (Section \ref{sec:memoryless}), while if we are to exclude previously tested elements, we may set $s=2^N$ corresponding to an $N$-bit classical memory (Section \ref{sec:withmemory}).

The set of allowed elementary operations to be applied between two noise hits is
\begin{equation}\label{specificationSgroverAppendix}
S~=~\left\{\,T\,\,\Big|\,\,\exists\, C_1,C_2\in\CPT(\B(\mathcal{H})\otimes\mathbb{C}^s\otimes\mathbb{C}^{O}\otimes\mathbb{C}^2):\,\forall x\in X:\,T(x)=C_2\circ\widehat{O}_x\circ C_1\,\right\}~,
\end{equation}
where $\widehat{O}_x$ first prepares the pure state $\frac{1}{\sqrt{2}}(\ket{0}-\ket{1})$ on the $\mathbb{C}^2$-subsystem of $\mathcal{H}$, and then acts as $\ket{x',b}\mapsto\ket{x',b+\delta_{x,x'}}$ on $\mathcal{H}$ (cf.\ Eqs.\ (\ref{oraclebasic}) and (\ref{oracleaction})) and as the identity on the classical registers. This choice of $S$ means that an elementary step consists of a single use of the oracle, possibly applying an arbitrary (but oracle-independent!) channel before and afterwards.

Finally, half of the trace-distance in (\ref{defineminimalsteps}) gives, when the ready-flag $\mathbb{C}^2$ has been set to $1$, exactly the probability of not outputting the correct oracle index in the classical output register, and it is this failure probability which we wanted to be smaller than $\varepsilon$ in Sections \ref{sec:memoryless} and \ref{sec:withmemory}.
\end{example}

\begin{example}[Binary addition]\label{binaryadditionexample}This example illustrates the possibility to introduce some kind of ``locality structure''. The task is the addition of two $n$-bit numbers given their binary representation using local gates on a $2n$-bit quantum register with local dephasing noise. Such a task is given by $X=\{0\}$, $I=\{0,1,\ldots,2^n-1\}\times\{0,1,\ldots,2^n-1\}$, $O=\{0,1,\ldots,2^{n+1}-1\}\simeq\{0,1\}^{n+1}$, $f_{x(i_1,i_2)o}=\delta_{i_1+i_2,o}$, the Hilbert space is $\mathcal{H}=\left({\mathbb{C}^2}\right)^{\otimes 2n}$, $\varrho_{(i_1,i_2)}=\ket{i_1,i_2}$ (with $i_1,i_2$ considered as a $2n$-bit string), $D_p=d_p^{\otimes 2n}$ with $d_p:\B(\mathbb{C}^2)\to\B(\mathbb{C}^2)$ the partial dephasing with strength $p$, $s=1$, and $S\subseteq\CPT(\B({\mathbb{C}^2}^{\otimes 2n})\otimes\mathbb{C}^{O}\otimes\mathbb{C}^2)$ consisting of $1$- and $2$-(qu)bit gates, i.e.~channels which act as the identity on all but at most two bits (quantum or classical), remembering the subsystem structure of $\mathbb{C}^O\simeq\left(\mathbb{C}^2\right)^{\otimes n+1}$.
\end{example}

An algorithm that works only for \emph{known} noise level $p$ is \emph{not} fault-ignorant; such algorithms may be formalized by assuming a $p$-dependence $(T^{\varepsilon,p}_n)_n$ in the family of sequences in Definition \ref{definefaultignorantalgo}. On the other hand, if these sequences do \emph{not} depend on the desired accuracy $\varepsilon$, i.e.~$(T^{\varepsilon}_n)_n\equiv(T_n)_n$, then the algorithm does have another feature:~the level of accuracy $\varepsilon$ need not be specified before starting the algorithm; when higher accuracy is desired (i.e.\ smaller $\varepsilon$), one only needs to continue running the algorithm for longer time.

Returning to efficiency questions, one may call a fault-ignorant algorithm (or rather, a family of fault-ignorant algorithms, parametrized by some ``problem size'' $N$) \emph{efficient} if, for any $p$, any $\varepsilon$ and any $N$, its runtime is within a constant factor times the runtime of the best algorithm that may depend on $\varepsilon$ \emph{and} on $p$ (see Section \ref{moreformaldescription}). In this sense, our Theorems \ref{thmsymmalgomemoryless} and \ref{thmlowerboundwithmemory} can be seen as statements that Algorithms \ref{algounknownpnotexcluding} and \ref{algounknownpexcluding} are efficient (within restricted classes of algorithms).

It should be clear that there is nothing special about the set $[0,1]$ parametrizing the noise channels apart from the possibility to interpret it as ``strength'' or to use it directly as a coefficient in a convex combination. One could instead consider a family $(D_p)_{p\in P}$ of noise channels indexed by an arbitrary set $P$ parametrizing wider classes of noise, and so allowing for ``more'' ignorance about the faults. Another obvious extension of Definition \ref{definitionnoisytask} would be to allow for time-dependent noise.

\section{Noise models}\label{sec:noisemodels}
Here we elaborate on different kinds of noise which may be acting on the quantum computer during its runtime, and in particular on the noise models to which our results apply.

\emph{Partial depolarizing},
\begin{equation}
D_p^{I}(\varrho)~=~p\frac{I_d}{d}\Tr(\varrho)\,+\,(1-p)\varrho\label{partialdepolsecond}
\end{equation}
for noise level $p\in[0,1]$, has been defined in Eq.~(\ref{partiallydepolnoise}), and corresponds to erasing the state of the quantum register with probability $p$ (between any two oracle calls). Somewhat similar is \emph{partial dephasing},
\begin{equation}
D_p^{\varphi}(\varrho)~:=~p\sum_{x=1}^d\bra{x}\varrho\ket{x}\ketbra{x}{x}\,+\,(1-p)\varrho\label{definepartialdephasing}~,
\end{equation}
acting on states $\varrho$ on a $d$-dimensional Hilbert space equipped with a distinguished orthonormal basis $\{\ket{x}\}_{x=1}^d$ (for these, we imagine the basis states with respect to which the oracles (\ref{oraclebasic}) act). For $p=1$, all quantum coherence is lost between any two oracle calls, but one can still perform a classical algorithm (on the basis states $\ket{x}$); in this sense, the noise level $p$ of partial dephasing parametrizes how ``quantum'' a search algorithm may be. Our constructive algorithms also work with the runtimes guaranteed by Theorems \ref{thmunknownpwithoutexcluding} and \ref{thmknownpwithoutexcluding} under the more general noise model
\begin{equation}
D_p^{T}(\varrho)~:=~pT(\varrho)\,+\,(1-p)\varrho~,\label{gentnoise}
\end{equation}
where $T$ may be any quantum channel, see discussion below Eq.\ (\ref{partiallydepolnoise}).

Our formalization of noisy search algorithms (Sections \ref{sec:memoryless} and \ref{sec:withmemory}) does allow to \emph{noiselessly check} whether a given index $x'$ equals the marked element $x$, since immediately before and after an oracle call one may perform any quantum operation without noise (cf.~Eq.~(\ref{statefrombasicbuildingblock})) and thus one action of $\widetilde{O}_x$ from (\ref{oraclebasic}) on a suitably prepared quantum register can accomplish this check and write the result into the (noiseless) classical memory. This fact is important, as it allows the verification/falsification step at the conclusion of each round (cf.\ Algorithms \ref{algounknownpnotexcluding} and \ref{algounknownpexcluding}). Alternatively, such a noiseless check may be implemented by a classical table lookup.

The above noise models are formulated in discrete time, but our prescription for the noise $D_p$ to act between any two oracle calls is supposed to model the continuous action of noise in a real-world situation. For example, since in (\ref{partialdepolsecond}) the probability to ``lose'' the quantum computer between any two consecutive oracle calls is $p$, its lifetime is roughly $1/p$ (measured in the time between two oracle calls); and indeed, the time scale $k\sim1/p$ appeared often in the analysis in Subsection \ref{subsectrepeatbuildingblockknownp}.

Note that quantum error correction \cite{shorfaulttolerant,klz98b,nielsenchuang} does not work for partial depolarizing (\ref{partialdepolsecond}) or dephasing (\ref{definepartialdephasing}), as these noises affect the whole quantum computer ``collectively''. This means that the whole quantum computer is subjected to a ``flash'' of noise, such as drifting lasers or an external hit by a magnetic field. These may be reasonable noise models for not-too-large quantum computers.

\bigskip

Discussing the noise models more quantitatively, we first notice that the lower bound (\ref{lowerboundprobs}) on the success probability after $k$ steps of Grover's algorithm under noise applies to all three noise models (\ref{partialdepolsecond})--(\ref{gentnoise}). For partial depolarizing (\ref{partialdepolsecond}) one can compute the success probability in (\ref{lowerboundprobs}) exactly: the $2^k-1$ omitted terms are of the form
\begin{equation}
\sum_{x=1}^N\frac{1}{N}\,p^m(1-p)^{k-m}\bra{x}\frac{I_N}{N}\ket{x}~=~\frac{1}{N}\,p^m(1-p)^{k-m}~,
\end{equation}
where $m$ is the number of noise hits. These terms correspond to events when the maximally mixed state is prepared at some point due to noise acting and since both $D_p$ and $G_x$ (see before Eq.~(\ref{firststatementsofpsinmaintext})) are unital. As the coefficients of these terms sum up to $(1-(1-p)^k)$, the exact success probability for this model is
\begin{equation}\label{exactsuccessprobabfordepol}
p_s^{pol}(N,k,p)~=~\left(1-(1-p)^k\right)\frac{1}{N}\,+\,(1-p)^k\sin^2\left((2k+1)\arcsin{\frac{1}{\sqrt{N}}}\right)~.
\end{equation}
Using this exact success probability for partial depolarizing improves the runtime bounds for this specific noise model (e.g.~Theorem \ref{thmunknownpwithoutexcluding} for large noise level $p$), but the lower bound (\ref{lowerboundprobs}) is quite tight unless $kp\gg1$. The drawbacks of relying on a too specific noise model are furthermore discussed below Eq.\ (\ref{eq:firstbound}).

An exact computation of the success probability can also be done for partial dephasing (\ref{definepartialdephasing}), but is much more involved. Furthermore, one can prove that the success probability for partial dephasing is not smaller than for depolarizing at the \emph{same} noise level:~$p_s^{\varphi}(N,k,p)\ge p_s^{pol}(N,k,p)$. This inequality is, however, not immediate, as for identical noise parameters $p\in(0,1)$, partial depolarizing $D_p^{I}$ \emph{cannot} be obtained by post-processing $D_p^{\varphi}$, i.e.~$D_p^{I}\neq P\circ D_p^{\varphi}$ for all quantum channels $P$.

Our proofs of the general lower bounds on the number of oracle calls (Theorems \ref{thmsymmalgomemoryless} and \ref{thmlowerboundwithmemory}) require partial depolarizing (\ref{partialdepolsecond}), as Theorem \ref{thm8kplus1overp} was proved only for generalized partial depolarizing noise $\varrho\mapsto p\tau\Tr{\varrho}+(1-p)\varrho$ and the proofs (and presuppositions) of Theorems \ref{thmsymmalgomemoryless} and \ref{thmlowerboundwithmemory} require furthermore a symmetry between the oracle indices, limiting further to $\tau=I_d/d$.

\bigskip

Finally, we argue that it makes sense in Sections \ref{sec:memoryless} and \ref{sec:withmemory} to perform efficiency analyses by keeping the noise parameter $p$ fixed while the size of the quantum register $N$ (or $NM$) varies, possibly even tending to infinity. Phrased another way, we ask whether, for example, the strength of partial depolarizing (\ref{partialdepolsecond}) with parameter $p$ on an $d$-dimensional quantum system is comparable to the strength of partial depolarizing with the same parameter $p$ in $d'$ dimensions, even if $d$ and $d'$ are widely different.

First, both partial depolarizing (\ref{partialdepolsecond}) and partial dephasing (\ref{definepartialdephasing}) (the latter with respect to a tensor product of bases) are compatible under tracing out subsystems when the \emph{same} parameter $p$ is used on the tensor product system and on the subsystem:
\begin{equation}
{\rm tr}_B\left[D^{I,\varphi}_p(\varrho_{AB})\right]~=~D^{I,\varphi}_p\left({\rm tr}_B[\varrho_{AB}]\right)\qquad\forall p\in[0,1]~.
\end{equation}
With this parametrization of the noise, it does therefore not help for algorithm performance to introduce larger and larger ancillary systems or ``innocent bystanders'':~the noise on the ``Grover part'' of the algorithm cannot be made small in this way, which is a reasonable requirement.

Secondly, both for partial depolarizing and dephasing, one can obtain the noise in $d-1$ dimensions by post-processing the noise on a $d$-dimensional system:
\begin{equation}
D_p^\varphi(\varrho)\oplus0~=~P\circ D_p^\varphi(\varrho\oplus0)\qquad\forall p\in[0,1]\,\forall\varrho\in\B(\mathbb{C}^{d-1})~,
\end{equation}
where $P={\rm{id}}$ for dephasing noise (and the additional dimension $\ket{d}$ has to correspond to one of the basis vectors in (\ref{definepartialdephasing})), and $P(X):=(I-\ket{d}\bra{d})X(I-\ket{d}\bra{d})+(I-\ket{d}\bra{d})\bra{d}X\ket{d}/(d-1)$ for depolarizing. This compatibility under restrictions of the Hilbert space to subspaces is important and sensible in the context of exclusion algorithms (Algorithm \ref{algounknownpexcluding}, and proof of Theorem \ref{thmknownpwithoutexcluding}), where the effective dimension of the quantum register is reduced by $1$ in each round.

\section{Proof of Theorem \ref{thmunknownpwithoutexcluding}}\label{mainthm}
\begin{proof}As we assume the noise to act symmetrically with respect to the different oracles (which both partial depolarizing and dephasing do) and since the Grover steps of Algorithm~\ref{algounknownpnotexcluding} are symmetric as well, the success events in different rounds $g$ are independent. Thus, with (\ref{lowerboundprobs}), we can upper bound the failure probability after round $g^*$ by
\begin{equation}\label{repeatfailureprobabinappendix}
\begin{split}
\fail~\equiv~\prod_{g=0}^{g^*}\left(1-p_s(N,k_g,p)\right)~
&=~\exp\left\{\sum_{g=0}^{g^*}\log\left(1-p_s(N,k_g,p)\right)\right\}~\le~\exp\left\{-\sum_{g=0}^{g^*}p_s(N,k_g,p)\right\}\\
&\le~\exp\left\{-\sum_{g=0}^{g^*}(1-p)^{k_g}\sin^2\left((2k_g+1)\arcsin\frac{1}{\sqrt{N}}\right)\right\}~.
\end{split}
\end{equation}
To show that the failure probability is at most $\varepsilon$, as desired, below we will lower bound the sum $\sum_{g=0}^{g^*}$ and adjust parameters such that it is at least $\log(1/\varepsilon)$. The sum can be further bounded by assuming it to start at some $g=g_*$ with $0\leq g_*\leq g^*$:
\begin{equation}
\label{eq:sumfail}
\sum_{g=0}^{g^*}\left(\ldots\right)~\geq~(1-p)^{k_{g_*}}\sum_{g=g_*}^{g^*}\sin^2\left((2k_g+1)\arcsin\frac{1}{\sqrt{N}}\right)~.
\end{equation}

The number of oracle calls can be upper bounded as follows:
\begin{equation}\label{repeatupperboundonoraclecallsinappendix}
\begin{split}
\sum_{g=0}^{g^*}(k_g+1)~
&\le~g^*+1+\frac{\pi}{4}\sqrt{N}\,\sum_{g=0}^{g^*}\left(1+\frac{g}{c\log(1/\varepsilon)}\right)^{-1/2}\\
&\le~g^*+1+\frac{\pi}{4}\sqrt{N}+\frac{\pi}{4}\sqrt{N}\int_{g=0}^{g^*}\left(1+\frac{g}{c\log(1/\varepsilon)}\right)^{-1/2}dg\\
&\le~g^*+1+\frac{\pi}{4}\sqrt{N}+\frac{\pi}{2}\sqrt{N}\left(c\log\frac{1}{\varepsilon}\right)\sqrt{1+\frac{g^*}{c\log(1/\varepsilon)}}~.
\end{split}
\end{equation}

The proof of the theorem now consists in showing that there exists a number $g^*$ (of rounds) such that the failure probability (\ref{repeatfailureprobabinappendix}) is at most $\varepsilon$, while the number of oracle calls (\ref{repeatupperboundonoraclecallsinappendix}) does not exceed the value given in (\ref{stepsinfirstfaulttolerantalgo}). This argument will be split into three cases, as sketched in the main text. We make abundant use of the fact that $c\log(1/\varepsilon)\geq 1$, since $\varepsilon\leq\varepsilon_0:=1/2$ and $c=10$. We also define $N_0:=100$ and assume $N\geq N_0$ throughout, in accord with the statement of Theorem \ref{thmunknownpwithoutexcluding}.

\medskip

{\bf{Case 1:}} $p\leq4/(\pi\sqrt{N})$. In this case, the actual decoherence rate $p$ is small, and we take only the first few rounds $g$ into account to obtain an upper bound on $\fail$. By using $\arcsin(1/\sqrt{N})\geq 1/\sqrt{N}$, $\sin(x)\ge x\sin(x_0)/x_0$ for $0\leq x\le x_0\le\pi$, and setting
\[
Q(N_0)~:=~\frac{\sin^2\left[\left(\frac{\pi}{2}\sqrt{N_0}+1\right)\arcsin(1/\sqrt{N_0})\right]}{\left[\left(\frac{\pi}{2}\sqrt{N_0}+1\right)\arcsin(1/\sqrt{N_0})\right]^2}~,
\]
we continue in bounding (\ref{eq:sumfail}):
\begin{equation}
\label{eq:sumfailadvanced}
\begin{split}
\sum_{g=0}^{g^*}\left(\ldots\right)~&\geq~(1-p)^{k_{g_*}}Q(N_0)\sum_{g=g_*}^{g^*}\left(2\lfloor\alpha_g(\varepsilon)\frac{\pi}{4}\sqrt{N}\rfloor+1\right)^2\frac{1}{N}\\
&\geq~\frac{\pi^2}{16}(1-p)^{k_{g_*}}Q(N_0)\sum_{g=g_*}^{g^*}\alpha^2_g(\varepsilon)~\geq~\frac{\pi^2}{16}(1-p)^{k_{g_*}}Q(N_0)\int_{g=g_*}^{g^*}\alpha^2_g(\varepsilon)dg~,
\end{split}
\end{equation}
where we used $2\lfloor x\rfloor+1\geq x$ for $x\geq 0$.
Choosing $g_*:=0$ and $g^*:=\left\lceil c_0\log(1/\varepsilon)\right\rceil$, we evaluate the integral in (\ref{eq:sumfailadvanced}) to find
\begin{equation}
\begin{split}
\sum_{g=0}^{g^*}\left(\ldots\right)~
&\geq~\frac{\pi^2}{16}(1-p)^{k_0}Q(N_0)c\log\left(1+\frac{c_0}{c}\right)\log\frac{1}{\varepsilon}\\
&\geq~\frac{\pi^2}{16}\left(1-\frac{4}{\pi}\frac{1}{\sqrt{N_0}}\right)^{\frac{\pi}{4}\sqrt{N_0}}Q(N_0)c\log\left(1+\frac{c_0}{c}\right)\log\frac{1}{\varepsilon}~,
\end{split}
\end{equation}
where we used $p\leq\pi/(4\sqrt{N})$. When the prefactor of $\log(1/\varepsilon)$ is at least $1$, then the failure probability $\fail$ will be at most $\varepsilon$ by Eq.~(\ref{repeatfailureprobabinappendix}); this happens e.g.~for the choice $c_0:=170$. The number of oracle calls from (\ref{repeatupperboundonoraclecallsinappendix}) is then
\begin{equation}\label{upperboundonruntimeincase1}
\leq~2+\frac{\pi}{4}\sqrt{N}+\left(\frac{\pi}{2}\sqrt{2c^2+c_0c}+\frac{c_0}{\sqrt{N_0}}\right)\sqrt{N}\log\frac{1}{\varepsilon}~,
\end{equation}
which is less than $\sqrt{N}+86\sqrt{N}\log(1/\varepsilon)$ due to $2+\pi\sqrt{N}/4\leq\sqrt{N}$. We notice that the term linear in $N$ (cf.~Eq.~(\ref{stepsinfirstfaulttolerantalgo})) is absent from the runtime (\ref{upperboundonruntimeincase1}) in Case 1; intuitively speaking, such small noise levels $p$ still allow for quadratic speedup in the quantum search.

{\bf{Case 2:}} $4/(\pi\sqrt{N})\leq p\leq p^*$, where we define $p^*:=0.3$. In this intermediate region of the actual decoherence rate (the need for $p^*<1$ will become evident later), we define $g^*:=\left\lceil c_2(\pi^2/16)Np^2\log(1/\varepsilon)\right\rceil$  and $g_*:=\left\lceil c_1(\pi^2/16)Np^2\log(1/\varepsilon)\right\rceil$ with $c_2>c_1>c$ to be determined later. Our choice $c_1>c$ will in particular imply $g_*\geq c\log(1/\varepsilon)$, so that we can continue lower-bounding (\ref{eq:sumfailadvanced}) by bounding the integrand,
\begin{equation}
\begin{split}
\sum_{g=0}^{g^*}\left(\ldots\right)~
&\geq~\frac{\pi^2}{32}(1-p)^{k_{g_*}}Q(N_0)\left(c\log\frac{1}{\varepsilon}\right)\int_{g=g_*}^{g^*}\frac{dg}{g}~\geq~\frac{\pi^2}{32}(1-p)^{k_{g_*}}Q(N_0)\,c\log\frac{g^*}{g_*}\log\frac{1}{\varepsilon}\\
&\geq~\frac{\pi^2}{32}\exp\left\{\frac{\pi}{4}\sqrt{N}\left(1+\frac{g_*}{c\log(1/\varepsilon)}\right)^{-1/2}\log(1-p)\right\}Q(N_0)\,c\log\frac{c_2}{c_1+\left(\frac{\pi^2}{16}Np^2\log(1/\varepsilon)\right)^{-1}}\log\frac{1}{\varepsilon}\\
&\geq~\frac{\pi^2}{32}\exp\left\{\frac{\log(1-p^*)}{p^*}\sqrt{\frac{c}{c_1}}\right\}Q(N_0)\,c\log\frac{c_2}{c_1+\log^{-1}(1/\varepsilon_0)}\log\frac{1}{\varepsilon}~,
\end{split}
\end{equation}
where we used $p\log(1-p^*)\le p^*\log(1-p)$ (due to $p\geq p^*$) and $(\pi^2/16)Np^2\log(1/\varepsilon)\geq\log(1/\varepsilon_0)$ (due to $p\geq4/(\pi\sqrt{N})$). Again, the prefactor of $\log(1/\varepsilon)$ can be made larger than $1$ by choosing $c_1:=20$ and $c_2:=180$.  The number of oracle calls from (\ref{repeatupperboundonoraclecallsinappendix}) is then
\begin{equation}
\leq~2+\frac{\pi}{4}\sqrt{N}+\frac{\pi}{\sqrt{2}}\sqrt{N}c\log\frac{1}{\varepsilon}+\frac{\pi^2}{16}\left(p^*c_2+2\sqrt{cc_2}\right)Np\log\frac{1}{\varepsilon}~,
\end{equation}
which again is less than $\sqrt{N}+86(Np+\sqrt{N})\log(1/\varepsilon)$.

{\bf{Case 3:}} $p\geq p^*$. For large actual noise levels $p$, it is enough to consider only those rounds $g$ for which $k_g(\varepsilon)=0$; in each such round, a measurement is performed on the equal superposition state, leading to a success probability of exactly $1/N$. This leads to the choice $g_*:=\lfloor c(\pi^2/16)N\log(1/\varepsilon)\rfloor$ and $g^*:=\left\lfloor c_3(\pi^2/16)N\log(1/\varepsilon)\right\rfloor$, and we can lower-bound (\ref{eq:sumfail}):
\begin{equation}
\sum_{g=0}^{g^*}\left(\ldots\right)~\geq~\sum_{g=g_*}^{g^*}\frac{1}{N}~=~\frac{1}{N}\left(g^*-g_*+1\right)~\geq~\frac{\pi^2}{16}(c_3-c)\log\frac{1}{\varepsilon}~.
\end{equation}
By choosing $c_3:=12$, the prefactor of $\log(1/\varepsilon)$ exceeds $1$.  The number of oracle calls from (\ref{repeatupperboundonoraclecallsinappendix}) is then
\begin{equation}
\leq~1+\frac{\pi}{4}\sqrt{N}+\frac{\pi}{2}c\sqrt{N}\log\frac{1}{\varepsilon}+\frac{\pi^2}{16p^*}\left(c_3+2\sqrt{cc_3}\right)Np\log\frac{1}{\varepsilon}~,
\end{equation}
which is again less than $\sqrt{N}+86(Np+\sqrt{N})\log(1/\varepsilon)$.

\medskip

So far we have proved that the algorithm is fault-tolerant with runtime at most $\sqrt{N}+86(Np+\sqrt{N})\log(1/\varepsilon)$. Due to $\varepsilon\leq\varepsilon_0=1/2$, (\ref{stepsinfirstfaulttolerantalgo}) is an upper bound on the runtime.
\end{proof}
For the proof of the prefactor $20$ mentioned below Theorem \ref{thmunknownpwithoutexcluding}, the Case 3 in the proof above can be neglected, and we set $p^*:=\epsilon_0:=0.1$, and alter the lower bound (\ref{eq:sumfailadvanced}) a bit, such that constants $c_0$, $c_1$, $c_2$ etc.\ can be found to yield the lower guaranteed runtime.

\section{Technical Lemmata}\label{app:technical}
By the following lemma, we convert Zalka's implicit bound, Eq.~(\ref{zalkaimplicit}), into an explicit one (see after the proof below):
\begin{lemma}\label{lemmaxyineq}
Let $0\le x\le 1$ and $0\le y\le 1$. Then, for any $0<\alpha<\infty$:
\begin{equation}
\sqrt{xy}+\sqrt{(1-x)(1-y)}~\le~1+\frac{1}{2}(\alpha-1)x+\frac{1}{2}\left(\frac{1}{\alpha}-1\right)y\label{lemmaxyineqeqn}~.
\end{equation}
\end{lemma}
\begin{proof}
The left hand side is a concave function of $(x,y)\in[0,1]^2$, smooth in the interior, and hence its graph stays under its tangent plane drawn at any point $(x_0,y_0)$ for $0<x_0<1$ and $0<y_0<1$. The partial derivatives of $h(x,y)=\sqrt{xy}+\sqrt{(1-x)(1-y)}$ at $(x,y)$ are
\begin{equation}
h_x(x,y)=\frac{\sqrt{y}}{2\sqrt{x}}-\frac{\sqrt{1-y}}{2\sqrt{1-x}}\quad\textrm{and}\quad h_y(x,y)=\frac{\sqrt{x}}{2\sqrt{y}}-\frac{\sqrt{1-x}}{2\sqrt{1-y}}~.
\end{equation}
Writing $y_0=\alpha^2x_0$ for $0<\alpha<\infty$, we have
\begin{equation}
\begin{split}
h(x,y)~
  &  \le~ h(x_0,y_0)+h_x(x_0,y_0)(x-x_0)+h_y(x_0,y_0)(y-y_0)  \\
  & =~h(x_0,\alpha^2x_0)+\frac{1}{2}\left(\alpha-\sqrt{\frac{1-\alpha^2x_0}{1-x_0}}\right)(x-x_0)+\frac{1}{2}\left(\frac{1}{\alpha}-\sqrt{\frac{1- x_0}{1-\alpha^2x_0}}\right)(y-\alpha^2x_0)~.
\end{split}
\end{equation}
Now taking the limit $x_0\to 0$ yields (\ref{lemmaxyineqeqn}).
\end{proof}
We apply this lemma to Zalka's bound \cite{Zalka} (Eq.~(\ref{zalkaimplicit}) above) with $x=p_s$ and $y=\frac{1}{N}$:
\begin{equation}
\begin{split}
4k^2~
  & \ge~2N-2\sqrt{N}\sqrt{p_s}-2\sqrt{N}\sqrt{N-1}\sqrt{1-p_s}  \\
  & =~2N\left[1-\left(\sqrt{p_s\frac{1}{N}}+\sqrt{(1-p_s)\left(1-\frac{1}{N}\right)}\right)\right]  \\
  & \ge~2N\left[1-\left(1+\frac{1}{2}(\alpha-1) p_s+\frac{1}{2}\left(\frac{1}{\alpha}-1\right)\frac{1}{N}\right)\right]  \\
  & =~(1-\alpha)Np_s-\left(\frac{1}{\alpha}-1\right)~.
\end{split}
\end{equation}
One can easily see that the sharpest bound on $p_s$ is obtained for $\alpha=(2k+1)^{-1}$, yielding $Np_s\le(2k+1)^2$.

\bigskip

The following lemma shows that Algorithm \ref{algounknownpexcluding} in Section \ref{sec:withmemory} is optimal within the class of algorithms considered in Theorem \ref{thmlowerboundwithmemory}, up to a constant factor in the runtime:
\begin{lemma}\label{memoryoptimaliylemma}
For $0<p<1$ and $0<\varepsilon<1$ the following inequalities hold:
\begin{equation}
\frac{1}{1+\frac{1}{p\log(1/\varepsilon)}}~\le~\min\left\{1-\varepsilon,\,\,p\log\frac{1}{\varepsilon}\right\}~\le~\frac{2}{1+\frac{1}{p\log(1/\varepsilon)}}~.
\end{equation}
\end{lemma}
\begin{proof}
First, $p\log(1/\varepsilon)>0$ implies
\begin{equation}
\frac{1}{1+\frac{1}{p\log(1/\varepsilon)}}~\le~p\log\frac{1}{\varepsilon}~.
\end{equation}

Now let $h(x)=x\log\frac{1}{x}$ for $x>0$. Then $h'(x)=-1-\log x$ and $h''(x)=-1/x$, so $h$ is concave, and the tangent at $x=1$ is $1-x$. This gives $\varepsilon p\log(1/\varepsilon)\le\varepsilon\log(1/\varepsilon)\le1-\varepsilon$, which implies
\begin{equation}
\frac{1}{1+\frac{1}{p\log(1/\varepsilon)}}~\le~1-\varepsilon~,
\end{equation}
concluding the left inequality. For the right inequality, if $\min\left\{1-\varepsilon,\,p\log\frac{1}{\varepsilon}\right\}=1-\varepsilon\le1$, then
\begin{equation}
\frac{2}{1+\frac{1}{p \log(1/\epsilon)}}~\geq~\frac{2}{1+\frac{1}{1-\epsilon}}~\geq~ 1-\epsilon~.
\end{equation}
Lastly, if $\min\left\{1-\varepsilon,\,p\log\frac{1}{\varepsilon}\right\}=p\log\frac{1}{\varepsilon}$ then in particular $0<p\log\frac{1}{\varepsilon}\le 1$. Thus, finally,
\begin{equation}
p\log\frac{1}{\varepsilon}~\le~\frac{2}{1+\frac{1}{p\log(1/\varepsilon)}}~.
\end{equation}
\end{proof}

\end{document}